\documentclass[letterpaper,pdftex,rmp,twocolumn,amsmath,amssymb,groupedaddress,floatfix]{revtex4}
\usepackage{graphicx}
\RequirePackage[sort&compress]{natbib}
\usepackage{natbib}
\usepackage{color}
\usepackage{textcomp}
\usepackage{makecell}
\usepackage{mathptmx}
\usepackage{float}
\usepackage{setspace}
\usepackage[a]{esvect}
\usepackage[font=sf,size=small,justification=RaggedRight,singlelinecheck=false]{caption}
\DeclareCaptionLabelSeparator{bar}{ \rule[-0.2em]{0.3ex}{1em} }
\usepackage[sf,bf,small]{titlesec}
\titlespacing*{\section}{0pt}{1em}{0em}

\definecolor{darkgray}{rgb}{0.25,0.25,0.25}
\definecolor{darkred}{rgb}{0.89,0.10,0.11}
\definecolor{darkblue}{rgb}{0.12,0.39,0.62}
\usepackage{url}

\usepackage[pdftex,breaklinks=true,colorlinks=true,citecolor=black,linkcolor=black,menucolor=black,urlcolor=darkblue,pdfborder={1 0 0}]{hyperref}
\hypersetup{pdftitle={Memory in network flows and its effects on spreading dynamics and community detection},pdfauthor={M. Rosvall, A. Viamontes Esquivel, A. Lancichinetti, J.D. West, and R. Lambiotte (2014)}}
\bibpunct{}{}{,}{s}{,}{,}

\usepackage{booktabs}

\newcommand{\mytoprule}{\specialrule{0.1em}{0em}{0em}}
\newcommand{\mybottomrule}{\specialrule{0.1em}{0em}{0em}} 
\newcommand{\mymidrule}{\specialrule{0.05em}{0em}{0em}}

\usepackage{dcolumn}
\usepackage{bm}
\usepackage{array} \newsavebox\cellbox
\newcolumntype{W}[2]
{>{\begin{lrbox}\cellbox}%
l%
<{\end{lrbox}%
\makebox[#2][#1]{\usebox\cellbox}}}
\usepackage{multirow}

\usepackage{cancel}	

\usepackage{amsthm,scrextend}
\usepackage{fancyhdr}

\newenvironment{proposition}[2][Proposition]{\begin{trivlist}
\item[\hskip \labelsep {\bfseries #1}{\bfseries #2.}]}{\end{trivlist}}
\begin{document}
\makeatletter
\renewcommand\@biblabel[1]{#1.}
\makeatother
	
\renewcommand{\figurename}{Figure}
\renewcommand{\thefigure}{\arabic{figure}}
\renewcommand{\tablename}{Table}
\renewcommand{\thetable}{\arabic{table}}
\renewcommand{\refname}{\large References}

\addtolength{\textheight}{1cm}
\addtolength{\textwidth}{1cm}
\addtolength{\hoffset}{-0.5cm}

\setlength{\belowcaptionskip}{1ex}
\setlength{\textfloatsep}{2ex}
\setlength{\dbltextfloatsep}{2ex}

\hyphenation{page-rank}

\newcommand*{\citen}[1]{%
  \begingroup
    \romannumeral-`\x 
    \setcitestyle{numbers}%
    \cite{#1}%
  \endgroup   
}

\newcommand{\enter}{\curvearrowleft}
\newcommand{\exit}{\curvearrowright}
\newcommand{\intra}{\circlearrowright}

\title{Mapping flows on sparse networks with missing links}

\author{Jelena Smiljani\'c}
\email{jelena.smiljanic@umu.se} 
 \affiliation{Integrated Science Lab, Department of Physics, Ume{\aa} University, SE-901 87 Ume{\aa}, Sweden} 
 \affiliation{Scientific Computing Laboratory, Center for the Study of Complex Systems, Institute of Physics Belgrade, University of Belgrade, Pregrevica 118, 11080 Belgrade, Serbia.}
\author{Daniel Edler}
 \affiliation{Integrated Science Lab, Department of Physics, Ume{\aa} University, SE-901 87 Ume{\aa}, Sweden}
 \affiliation{Gothenburg Global Biodiversity Centre, Box 461, SE-405 30 Gothenburg, Sweden.}
 \affiliation{Department of Biological and Environmental Sciences, University of Gothenburg, Carl Skottsbergs gata 22B, Gothenburg 41319, Sweden.}
\author{Martin Rosvall}
\affiliation{Integrated Science Lab, Department of Physics, Ume{\aa} University, SE-901 87 Ume{\aa}, Sweden}


\begin{abstract}

Unreliable network data can cause community-detection methods to overfit and highlight spurious structures with misleading information about the organization and function of complex systems. Here we show how to detect significant flow-based communities in sparse networks with missing links using the map equation. Since the map equation builds on Shannon entropy estimation, it assumes complete data such that analyzing undersampled networks can lead to overfitting. To overcome this problem, we incorporate a Bayesian approach with assumptions about network uncertainties into the map equation framework. Results in both synthetic and real-world networks show that the Bayesian estimate of the map equation provides a principled approach to revealing significant structures in undersampled networks.

\end{abstract}

\maketitle

\section{Introduction}

Unraveling the modular organization of social and biological systems with interactions comprising measured movements of some entity such as people, money, or information requires reliable maps of network flows~\cite{pons2005flows, rosvall2008maps, delvenne2010flows, schaub2012flows, lambiotte2014flows}. To find modular regularities in network flows, the map equation estimates a modular description length of the flows with information-theoretic measures. Optimizing the map equation with the search algorithm Infomap maximally compresses the modular description and detects significant flow-based communities when enough links are observed~\cite{rosvall2008maps, edler2017mapequation}. However, if too many links are missing, the map equation may highlight spurious communities resulting from mere noise. While there are generative methods that can deal with uncertain network structures, including link-prediction algorithms~\cite{guimera2009prediction, lu2011prediction, ghasemian2019overfitting} and network reconstruction approaches that often build on the stochastic block model~\cite{martin2016reconstruction, newman2018reconstruction, newman2018reconstruction2, peixoto2018reconstruction, squartini2018reconstruction}, no method can reliably identify flow-based communities in networks with missing links.

The map equation estimates the modular description length of network flows with the Shannon entropy~\cite{shannon1948mathematical}. With missing data, the Shannon entropy underestimates the actual entropy of the complete data~\cite{basharin1959entropy}. Consequently, when a network has many missing links, the map equation underestimates the actual description length of the complete network, capitalizes on details in the observed network, and favors network partitions with many small communities. While higher model complexity can further compress the description length, the resulting communities become sensitive to network perturbations. Having more missing links further obscures the community structure and leads to higher sensitivity. Overfitting happens when the communities poorly compress the description length of the complete network or other samples of the complete network~\cite{peixoto2014hierarchical,valles2018consistencies}.

Underestimating the entropy in networks with missing links also causes problems for standard procedures that evaluate model-prediction performance, including cross-validation: When the modular description length depends on the number of observed links, it also depends on the number of cross-validation folds such that only balanced but wasteful equal-sized splits of a network into training and test networks give useful results.

To overcome these problems, we present two regularization methods based on entropy estimation for undersampled discrete data. First, we incorporate a Bayesian approach in the map equation framework~\cite{wolpert1995bayes} and derive a closed-form formula for the posterior mean of the map equation under the Dirichlet prior distribution of network flows. Second, to enable more effective cross-validation, we measure the modular description length of the training and test networks for a given partition using Grassberger entropy estimation~\cite{grassberger2008entropy}. 

We show that the Bayesian estimate of the map equation does not detect spurious communities in the undersampled regime in either synthetic or real-world networks. Also, compared with the degree-corrected stochastic block model~\cite{peixoto2017dcsbm, peixoto2020mergesplit}, this approach gives solutions that are more robust to missing links in the analyzed networks. Moreover, with Grassberger entropy estimation, the modular description length becomes nearly independent of the amount of data: Instead of wasteful equal-sized splits, we can use most links in the training network to detect communities with Infomap and validate them using the remaining links in the test network. These two complementary solutions help us reduce overfitting and allow us to detect significant flow-based communities in networks with missing links.

\section{Mapping flows on complete networks}

The map equation is an information-theoretic objective function for community detection based on the equivalence between data compression and identifying regularities in data. Building on this minimum description length principle, the map equation estimates the per-step theoretical lower limit of the average code word length needed to describe network flows with a modular description~\cite{rosvall2008maps, edler2017mapequation}. When the links themselves do not represent flows, we can model the network flows with a random walker traversing the network. The goal is to identify the network partition that maximally compresses the modular description, which, at the same time, best captures the modular regularities of the network flows.

For simplicity, here we consider modular descriptions with a two-level community hierarchy (for the multilevel map equation, see Appendix~\ref{appendixB}). In a network with a well-defined community structure, the network flows stay for a relatively long time within communities. Therefore, to encode movements of the random walker between nodes with better compression, the map equation reuses short code words in modular codebooks instead of using unique code words for each node. For a uniquely decodable description, this approach requires an additional index codebook to encode transitions between communities.

The map equation measures the theoretical lower limit of the code length using the Shannon entropy \cite{shannon1948mathematical}. For partition $\mathsf{M}$ of nodes $\alpha=1 \ldots V$ in communities $i=1 \ldots m$, the map equation takes as input the probability that the random walker enters community $i$, $q_{i\enter}$, the probability to visit node $\alpha$, $p_{\alpha}$, and the probability to exit community $i$, $q_{i\exit}$. With $p_{i}^{\intra}=q_{i\exit}+\sum_{\alpha \in i}p_\alpha$ for the total use rate of module codebook $i$, the average per-step code length needed to describe random walker movements within community $i$ is
\begin{align}\label{eq:modulecodelength}
H(\mathcal{P}_{i}) = - \frac{q_{i\exit}}{p_{i}^{\intra}} \log_2\frac{q_{i\exit}}{p_{i}^{\intra}} 
 -\sum_{\alpha \in i} \frac{p_{\alpha}}{p_{i}^{\intra}}\log_2\frac{p_{\alpha}}{p_{i}^{\intra}}.
\end{align}
Similarly, the average per-step code length needed to describe random walker transitions between communities is
\begin{align}\label{eq:indexcodelength}
H(\mathcal{Q}) = - \sum_{i=1}^{m}\frac{q_{i\enter}}{q_{\enter}} \log_2\frac{q_{i\enter}}{q_{\enter}},
\end{align}
where $q_{\enter} = \sum_{i=1}^{m} q_{i\enter}$ is the total use rate of the index codebook.
Therefore, we can express the map equation as the sum of the average code length of all codebooks weighted by their use rate:
\begin{align}\label{eq:mapeq}
L(\mathsf{M}) = q_{\enter} H(\mathcal{Q}) + \sum_{i=1}^{m}p_{i}^{\intra}H(\mathcal{P}_i). 
\end{align}
To identify the partition that minimizes the map equation, Infomap explores the space of possible solutions in a stochastic and greedy fashion.

\section{Mapping flows on sparse networks with missing links}

Combined with Infomap, the map equation is an accurate method for community detection when complete network data are available~\cite{lancichinetti2009comparison}. However, empirical network data can lack data or contain measurement errors that cause missing or spurious links. When the map equation is applied to such unreliable network data, it may identify spurious communities with misleading information about the underlying network structure and function (Fig.~\ref{fig:network}).

\begin{figure}
 \includegraphics[width=8.6cm]{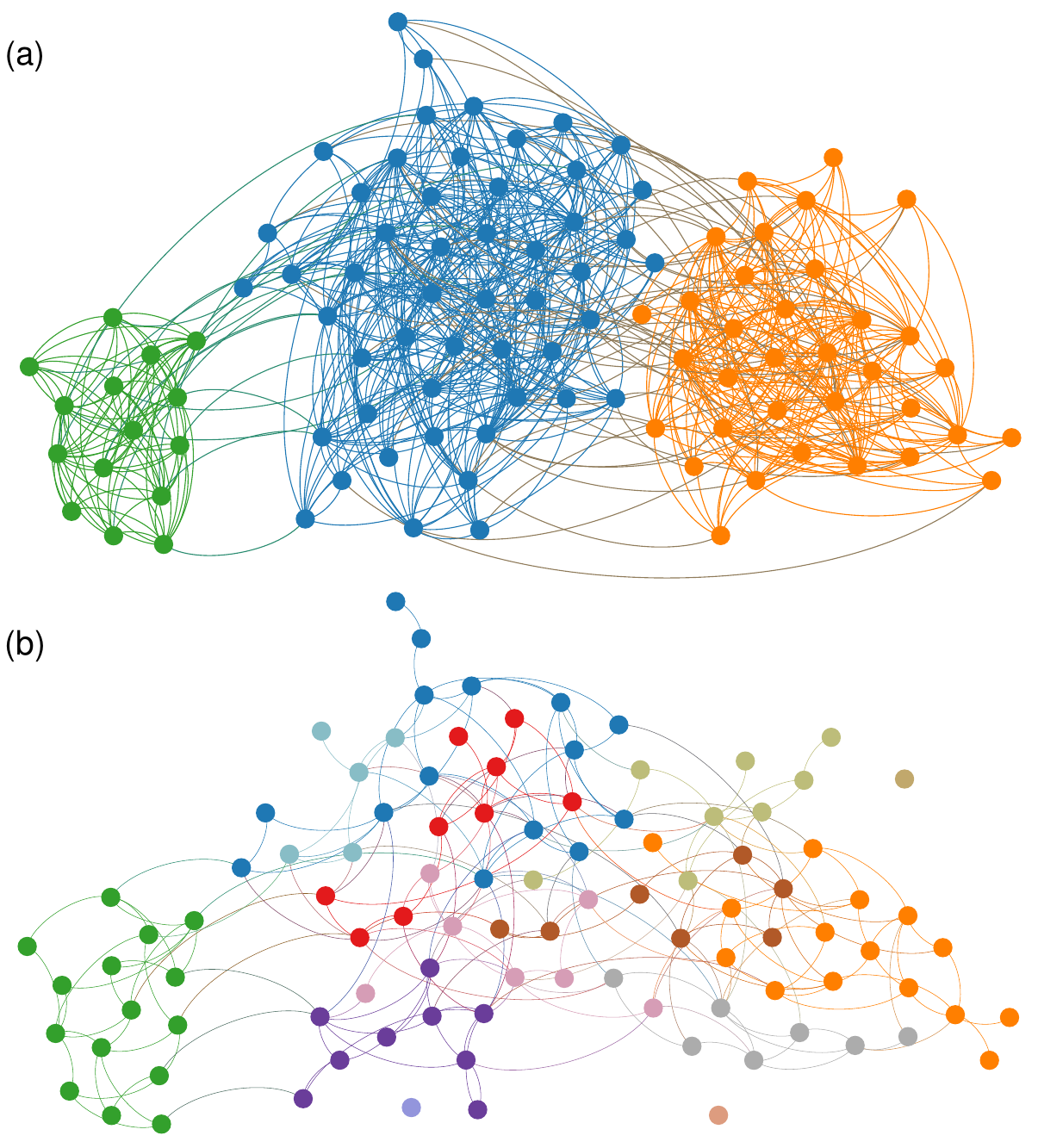} 
 \caption{\label{fig:network} Illustration of the overfitting problem in a small modular network. (a) The network has three communities. (b) When observing only a fraction of the links, the identified thirteen communities misrepresent the underlying network structure.}
\end{figure}

We focus on missing links, a common problem in social and biological networks, that causes the sample estimates of the random walker's transition probabilities to lose precision. When plugging the estimates into the Shannon entropy, the obtained entropy estimator suffers from a negative bias and underestimates the entropy terms of the map equation~\cite{basharin1959entropy}. Consequently, for the same partition $\mathsf{M}$, the description length decreases and the relative code length savings over the one-module solution, $l = 1-L(\mathsf{M})/L(\mathsf{1})$, increases with the number of missing links (Fig.~\ref{fig:crossvalidation}).

Worse yet, underestimating the index and module codebooks distorts their balance and shifts the optimal solution. The index codebook underrates the increase in between-module description length when using more communities, and the module codebooks overrate the within-module compression gain when using smaller communities. Also, stochastic fluctuations in missing links can lead the search algorithm off track because more undersampled regions attract community boundaries. Capitalizing on noise in this way underestimates not only the codebooks but also, primarily, the transition rates between communities. As a result, the map equation favors more and smaller communities in sparse networks with missing links~\cite{ghasemian2019overfitting} (Fig.~\ref{fig:network}). This effect is evident when so many links are missing that actual communities become sparse or even form disconnected components. Then the map equation cannot detect the actual communities; instead it overfits and identifies spurious communities from mere noise in the network. To overcome overfitting, we incorporate a Bayesian estimate of the map equation.

\begin{figure}
 \includegraphics[width=8.6cm]{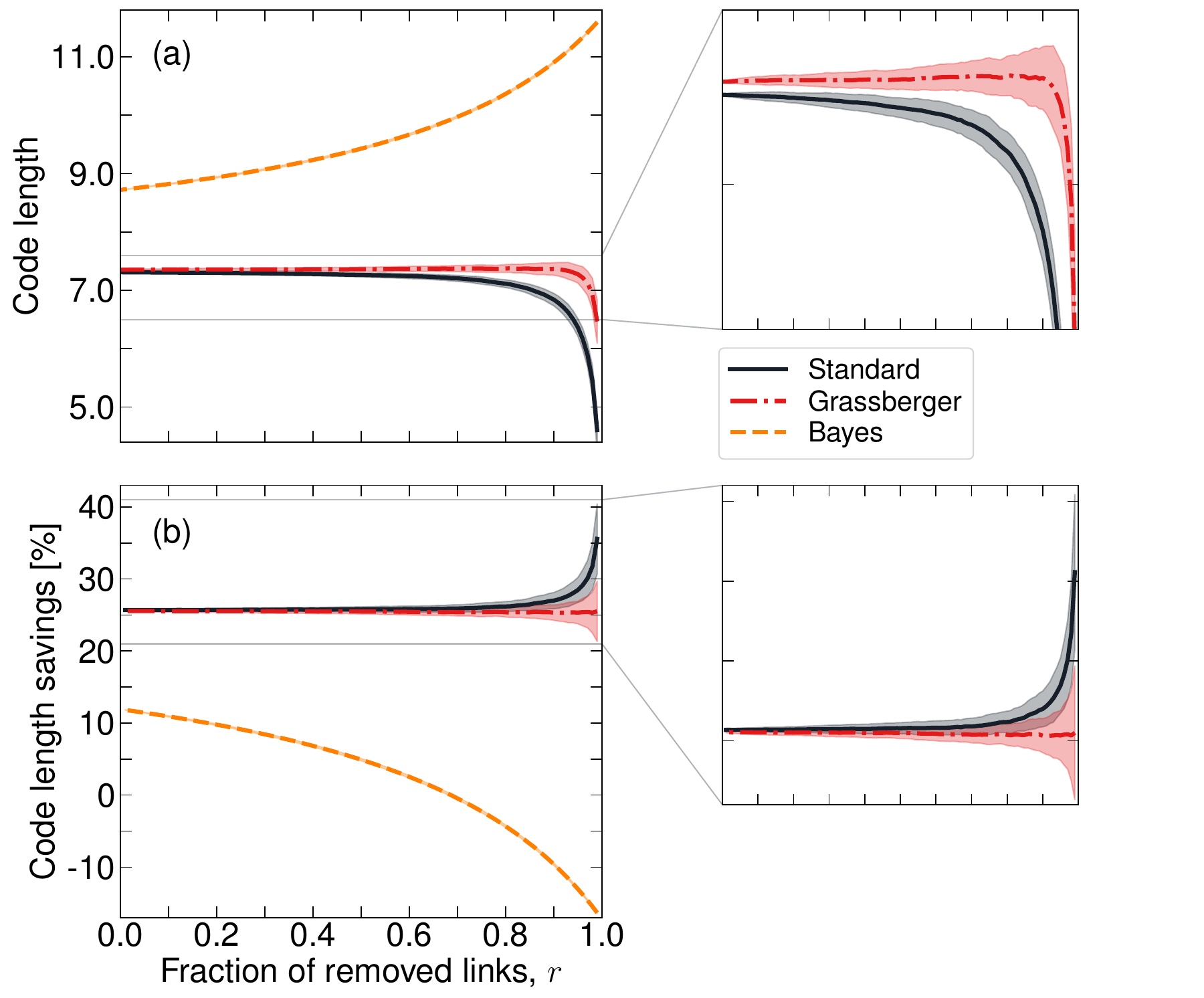}
 \caption{Modular compression in sparse networks. (a) Modular code length for planted partitions after link removal. (b) Relative code length savings in networks for planted partitions.
 With standard entropy estimation, the average description length decreases as the number of missing links increases, and we cannot compare relative code length savings in networks with different densities. In contrast, Grassberger entropy estimation almost eliminates the code length's density dependency. For $r>0.7$, the code length savings are negative for the Bayesian estimate of the map equation with prior $a_{\alpha} = \ln(V)$. By preferring the one-module solution over the planted partition in severely undersampled networks, the Bayesian estimate of the map equation avoids overfitting. For each $r$, we plot averages and variances over 100 network samplings of the synthetic network described in Sec.~\ref{sec:Results}. 
 }
 \label{fig:crossvalidation}
\end{figure}

\subsection{Bayesian estimate of the map equation}
\label{sec:bayesmapeq}
Different methods have been proposed to address the problem of entropy underestimation \cite{miller1955entropy, zahl1977jackknife, grassberger2008entropy, wolpert1995bayes, nemenman2002nsb, archer2014pym}. Methods based on bias reduction cannot prevent overfitting of the map equation because they have a high variance in the undersampled regime~\cite{miller1955entropy, zahl1977jackknife, grassberger2008entropy} and cannot deal with the underestimation of  the transition rates between communities. Instead, we use a Bayesian approach proposed by Wolpert and Wolf to estimate the function of probability distributions~\cite{wolpert1995bayes}. This method not only prevents overfitting to noisy structures better than other Bayesian estimators~\cite{nemenman2002nsb, archer2014pym}; it also enables an analytical estimation of the map equation and a computationally efficient implementation in Infomap.

In general, we seek the Bayesian estimator $\hat{f}_B$ of a function $f(\pmb{\rho})$ that takes a discrete probability distribution $\pmb{\rho} = (\rho_1,\rho_2,\dots,\rho_m)$ as input. When $\pmb{\rho}$ is not given and we have only observations $\pmb{n} = (n_1,n_2,\dots,n_m)$, with $\sum_{i=1}^{m}n_i=N$ sampled according to the distribution $\pmb{\rho}$ ($E(n_i) = \rho_i N$), we must estimate $f(\pmb{\rho})$ using the observed data $\pmb{n}$. The Bayesian estimator for $f(\pmb{\rho})$ is the posterior average, 
\begin{align}\label{eq:estimatorBayes}
 \hat{f}_B(\pmb{n}) = E[f|\pmb{n}] = \int f(\pmb{\rho}) P(\pmb{\rho}|\pmb{n})d\pmb{\rho},
\end{align}
where $P(\pmb{\rho}|\pmb{n})$ is the posterior over the unknown distribution $\pmb{\rho}$ given by Bayes' rule,
\begin{align}\label{eq:bayesrule}
 P(\pmb{\rho}|\pmb{n}) = \frac{P(\pmb{n}|\pmb{\rho})P(\pmb{\rho})}{P(\pmb{n})}.
\end{align}
To obtain $P(\pmb{\rho}|\pmb{n})$, we choose an appropriate prior probability distribution $P(\pmb{\rho})$ and use  the fact that the likelihood
\begin{align}
 P(\pmb{n}|\pmb{\rho}) = N! \prod_{i=1}^{m}\frac{\rho_{i}^{n_i}}{n_i!}
\end{align}
and the total probability of the data
\begin{align}
 P(\pmb{n}) = \int d\pmb{\rho} P(\pmb{n}|\pmb{\rho})P(\pmb{\rho}).
\end{align}

Applied to the map equation, we seek the Bayesian estimator of $f(\pmb{\rho})=L(\mathsf{M})$. Assuming undirected and unweighted links, the transition rate estimates are~\cite{mitzenmacher2005probability}:
\begin{align}
 p_{\alpha}=\frac{k_{\alpha}}{\sum_{\alpha=1}^{V}k_{\alpha}}, \label{eq:nodevisit}\\
 q_{i\enter} = \frac{k_{i\enter}}{\sum_{\alpha=1}^{V}k_{\alpha}}, \label{eq:moduleenter}\\
 q_{i\exit} = \frac{k_{i\exit}}{\sum_{\alpha=1}^{V}k_{\alpha}},\label{eq:moduleexit}
\end{align}
where $k_{\alpha}$ is the degree of node $\alpha$ and $k_{i\enter} = k_{i\exit}$ is the degree of module $i$, the number of links that connect nodes of module $i$ with nodes of other modules $j$, $j \neq i$. However, when the information about links is incomplete, the actual values of node and module degrees can deviate from these estimates. Therefore, we must apply a probabilistic approach, or the map equation will overfit and exploit spurious network structures.

To develop a Bayesian treatment of the map equation, for a given partition $M$, we specify a prior distribution $P(p_{\alpha}, q_{i\enter}, q_{i\exit})$ over the transition rates $p_{\alpha}, q_{i\enter}$, and $q_{i\exit}$. A convenient choice is the Dirichlet distribution, which has simple analytical properties and can be interpreted as a probability distribution over the multinomial distribution of the transition rates,
\begin{align}
\begin{split}
 &P(p_{\alpha}, q_{i\enter}, q_{i\exit}|a_{\alpha}, a_{i\enter}, a_{i\exit}) = \\
 &\frac{\Gamma(a_1+\dots+a_{m\exit})}{\Gamma(a_1) \dots \Gamma(a_{m\exit})} \prod_{\alpha=1}^{V}p_{\alpha}^{a_\alpha-1} \prod_{i=1}^{m}q_{i\enter}^{a_{i\enter}-1} \prod_{i=1}^{m}q_{i\exit}^{a_{i\exit}-1}.
\end{split}
\end{align}
Here $\Gamma(x)$ is the gamma function and $a_{1},\dots,a_{V}$, $a_{1\enter},\dots,a_{m\enter}$, and $a_{1\exit},\dots,a_{m\exit}$ are the parameters of the distribution. While $\sum_{\alpha=1}^{V}p_{\alpha} + \sum_{i=1}^{m}q_{i\enter} + \sum_{i=1}^{m}q_{i\exit} \neq 1$, we can use normalized transition rates because the map equation is scale invariant (see Appendix~\ref{appendixA}).

We obtain the posterior distribution of the transition rates in Eq.~(\ref{eq:bayesrule}) by multiplying the Dirichlet prior by the likelihood function and normalizing:
\begin{align}\label{eq:posteriorDirichlet}
\begin{split}
 P&(p_{\alpha}, q_{i\enter}, q_{i\exit}|k_{\alpha}, k_{i\enter}, k_{i\exit}, a_{\alpha}, a_{i\enter}, a_{i\exit}) \propto \\
 &\prod_{\alpha=1}^{V}p_{\alpha}^{k_\alpha + a_\alpha - 1} \prod_{i=1}^{m}q_{i\enter}^{k_{i\enter} + a_{i\enter} - 1} \prod_{i=1}^{m}q_{i\exit}^{k_{i\exit} + a_{i\exit} - 1}.
\end{split}
\end{align}
By combining this distribution and the expanded form of the map equation,
\begin{align}\label{eq:mapeqexpanded}
\begin{split}
 L&(\mathsf{M}) = -\sum_{\alpha=1}^{V}p_{\alpha}\log_2(p_{\alpha})-\sum_{i=1}^{m}q_{i\exit}\log_2(q_{i\exit})\\
       & +\sum_{i=1}^{m}\left(q_{i\exit}+\sum_{\alpha \in i} p_{\alpha}\right)\log_2\left(q_{i\exit}+\sum_{\alpha \in i} p_{\alpha}\right)\\
       & -\sum_{i=1}^{m}q_{i\enter}\log_2(q_{i\enter}) \\
       & + \left(\sum_{i=1}^{m}q_{i\enter}\right)\log_2\left(\sum_{i=1}^{m}q_{i\enter}\right),
\end{split}
\end{align}
in Eq.~(\ref{eq:estimatorBayes}), and integrating, we obtain a closed formula for the posterior average of the map equation,
\begin{align}\label{eq:estimatormapeq}
 \begin{split}
  \hat{L}_B&(\mathsf{M}) = \frac{1}{\ln(2)}\frac{1}{\sum_{\alpha=1}^{V}u_{\alpha}} \\
              & \times [ -\sum_{\alpha=1}^{V}u_{\alpha}\psi(u_{\alpha}+1) - \sum_{i=1}^{m}u_{i\exit}\psi(u_{i\exit}+1)\\
              & + \sum_{i=1}^{m}(u_{i\exit} + \sum_{\alpha \in i} u_{\alpha})\psi(u_{i\exit} + \sum_{\alpha \in i} u_{\alpha} + 1 ) \\  
              & - \sum_{i=1}^{m}u_{i\enter}\psi(u_{i\enter}+1)\\
              & + (\sum_{i=1}^{m}u_{i\enter})\psi(\sum_{i=1}^{m}u_{i\enter} + 1) ],
 \end{split}
\end{align}
where $u_x = k_x + a_x$ and $\psi(x)$ is the digamma function.

 The parameters $\pmb{a}$ reflect our prior assumption of the link distribution in the network before we observed the network data. After seeing the data, we update our assumption by increasing the value of $a_x$ by $k_x$ and obtain the posterior distribution. For a sparse, undersampled network, therefore, the prior parameters $\pmb{a}$ dominate the posterior link distribution. Conversely, as the network density increases, the posterior distribution becomes sharply peaked and the network data dominate the posterior link distribution. Proper selection of prior parameters $\pmb{a}$ is important for good performance.

We consider as an uninformative prior an Erd\H{o}s-R\'{e}nyi network with $V$ nodes, where each pair of nodes is connected with some constant probability $p$ \cite{erdos1959randomgraph}. The average degree is $\langle k \rangle = p V$ and sets the prior parameters to $a_{\alpha}=\langle k \rangle$ and $a_{i\enter}=a_{i\exit}=V_i(V-V_i)\frac{\langle k \rangle}{V-1}$, where $V_i$ is the number of nodes in module $i$. We aim to choose the average degree $\langle k \rangle$ such that the prior prevents the map equation from overfitting in the undersampled network, but also enables the map equation to detect well-formed communities. Since the random network experiences a phase transition from disconnected to connected at $\langle k \rangle = \ln(V)$ \cite{erdos1959randomgraph}, for $\langle k \rangle \ll \ln(V)$ the random network has isolated components and the prior cannot prevent overfitting, while for $\langle k \rangle \gg \ln(V)$ well-formed communities can merge such that the map equation underfits. At the phase transition between these extremes, $a \sim \ln(V)$ forms a principled prior.

Because there are no modular regularities in an Erd\H{o}s-R\'{e}nyi network, this choice of prior induces positive bias in the code length estimation (Fig.~\ref{fig:crossvalidation}(a)). When observing fewer links in a network, the prior network influences the posterior link distribution more such that the code length increases for the planted partition. Eventually, for severely undersampled networks, the Bayesian estimate of the map equation prefers the one-module solution and thereby avoids overfitting (Fig.~\ref{fig:crossvalidation}(b)).

This Bayesian estimate of the map equation extends to weighted networks where complete information about link weights is missing. If the link weights represent flows such that no flow modeling is necessary, the method also works for directed networks.

We have implemented the Bayesian estimate of the map equation in Infomap, available for anyone to use~\cite{infomap}. While we restrict our paper to the two-level formulation of the map equation for the sake of simplicity, the code also handles the Bayesian estimate of the multilevel map equation (see Appendix~\ref{appendixB}).

\subsection{The map equation with Grassberger entropy estimation}

An informative comparison between the standard map equation and a map equation with corrected entropy terms must take into account the structural properties of the detected communities. When possible, we can compare detected communities with planted communities; however, this approach does not work for real networks without known communities. To test for under- or overfitting in any network, we use cross-validation.

We first split the network data into training and test sets and apply Infomap to identify the partition that maximally compresses the description length of the training network. If Infomap successfully recovers a significant partition of the training network, the partition with maximal modular code length savings over the one-level code length will also successfully compress the description length of the test network. The opposite happens when there is not enough evidence in the data. Then Infomap overfits and detects a partition in the training network without code length savings in the test network. Thus, if Infomap detects a significant partition $\mathsf{M}$ without overfitting, the relative code length savings in the test network should be positive, $l^{test} = 1-L_{test}(\mathsf{M})/L_{test}(\mathsf{1}) > 0$ and close to the relative code length savings of the training network, $l^{test} \sim l^{train}$. Conversely, if Infomap overfits we expect $l^{test} < 0$.

However, the fact that the description length and the relative code length savings vary with the fraction of observed links limits the choice of training and test networks (Fig.~\ref{fig:crossvalidation}). Only with equal-sized training and test networks will the standard map equation underestimate their true description lengths to the same degree. But since equal splits waste half of the links on the test network, the training network of already sparse networks will be severely undersampled and possibly below the detectability limit. To reduce the description length's dependency on the fraction of observed links and enable effective cross-validation, we incorporate Grassberger entropy estimation~\cite{grassberger2008entropy} into the map equation.

For effective cross-validation, Grassberger entropy estimation enables the use of most of the links in the training network. We construct a test network by randomly removing a fraction $r$ of links from the network. The remaining links form a training network. With $E$ for the total number of links in the network and $k_{\alpha}$ for the degree of node $\alpha$, the probability that $k'_{\alpha}$ links of node $\alpha$ remain in the training network after removing $E-E'=rE$ links follows the hypergeometric distribution:
\begin{align}
 P(k'_{\alpha})=\frac{\binom{k_{\alpha}}{k'_{\alpha}}\binom{E-k_{\alpha}}{E'-k'_{\alpha}}}{\binom{E}{E'}}.
\end{align}
If $E$, $E'$, and $k_{\alpha}$ are sufficiently large, the hypergeometric distribution converges toward the Poisson distribution,
\begin{align}
 P(k'_{\alpha}) = \frac{\lambda^{k'_{\alpha}} }{k'_{\alpha}!}e^{-\lambda},
\end{align}
where the parameter $\lambda=\frac{E'k_{\alpha}}{E}=(1-r)k_{\alpha}$ such that $\langle k'_{\alpha} \rangle = (1-r)k_{\alpha}$.

For a given incomplete set of observations $(n_1,n_2,\dots,n_m)$, Grassberger entropy estimation assumes that they come from Poisson distributions with mean values $(z_1,z_2,\dots,z_m)$ and aims to construct a function $\phi(n)$ that minimizes the error $|z_{i}\ln(z_{i})-E(n_{i}\phi(n_i))|$ across all values of $z_i$~\cite{grassberger2008entropy}. The solution that minimizes the error is a recursive function $\phi(n)=G_n$ defined as
\begin{align}\label{eq:Gn}
\begin{split}
G&_1=-\gamma-\ln(2)\\
G&_2=2-\gamma-\ln(2)\\
G&_{2n+1}=G_{2n}\\
G&_{2n+2}=G_{2n}+\frac{2}{2n+1},
\end{split}
\end{align}
where $\gamma$ is Euler's constant~\cite{grassberger2008entropy}. 

While we cannot use Grassberger entropy estimation for weighted or directed networks, where visit rates correspond to the PageRank of the nodes~\cite{edler2017mapequation}, it does work for unweighted and undirected networks, where node visit and module transition rate estimates are given by link counts, Eqs.~(\ref{eq:nodevisit})--(\ref{eq:moduleexit}). Assuming incomplete observations, we can incorporate Grassberger entropy estimation into the map equation such that Eq.~(\ref{eq:mapeqexpanded}) takes the form
\begin{align}\label{eq:mapeqGrassberger}
 \begin{split}
  \hat{L}_G(\mathsf{M}) &= \frac{1}{\ln(2)}\frac{1}{\sum_{\alpha=1}^{V}k_{\alpha}} \\
              & \times [ -\sum_{\alpha=1}^{V}k_{\alpha}G_{k_{\alpha}} - \sum_{i=1}^{m}k_{i\exit}G_{k_{i\exit}}\\
              & + \sum_{i=1}^{m}(k_{i\exit} + \sum_{\alpha \in i} k_{\alpha})G_{k_{i\exit} + \sum_{\alpha \in i} k_{\alpha}} \\  
              & - \sum_{i=1}^{m}k_{i\enter}G_{k_{i\enter}} + (\sum_{i=1}^{m}k_{i\enter})G_{\sum_{i=1}^{m}k_{i\enter}} ].
 \end{split}
\end{align}
Grassberger entropy estimation also works for the multilevel formulation of the map equation~\cite{rosvall2011multilevelmapeq}.

Grassberger entropy estimation has high variance and low bias~\cite{schurmann2004bias}. Due to its high variance in the undersampled regime (Fig.~\ref{fig:crossvalidation}) and its lack of prior that can deal with underestimating the transition rates between communities, the map equation with Grassberger entropy estimation paired with Infomap does not perform better than the standard map equation on sparse networks with missing links. However, thanks to its low bias, the map equation with Grassberger entropy estimation applied to cross-validation with averaged code length over several network samplings can dramatically reduce the code length dependency on network density (Fig.~\ref{fig:crossvalidation}(a)). Also, for planted partitions, the average relative code length savings is practically independent of network density (Fig.~\ref{fig:crossvalidation}(b)). Consequently, we can use most links in the training network to reliably detect communities with Infomap.

\section{Results and discussion}\label{sec:Results}
We first analyze a synthetic network with planted community structure and a real-world Jazz collaboration network~\cite{gleiser2003jazznet}. We generate the synthetic network with the Lancichinetti-Fortunato-Radicchi (LFR) method~\cite{lancichinetti2008benchmark}. It has $V=1000$ nodes, average node degree $\langle k \rangle = 16$, and nodes partitioned into $M=35$ communities. The mixing parameter $\mu=0.3$ is the probability that a randomly chosen link will connect nodes from different communities. In the Jazz collaboration network, each node represents a band and two nodes are connected if there is at least one musician who has played in both bands. For this network with 198 nodes and 2,742 links, there is no information about ground-truth communities and no consensus about an optimal community partition \cite{newman2016resolution, peel2017groundtruth}. To generate sparse networks with missing links, we randomly remove a fraction $r$ of links from the networks, and average the results for each value of $r$ over 100 samplings.

Using these two networks, we compare the performance of the standard map equation, the Bayesian estimate of the map equation with different values of Dirichlet prior parameter $a_\alpha$, and the degree-corrected stochastic block model~\cite{peixoto2017dcsbm, peixoto2020mergesplit}. We are interested in the number of communities, the partition similarities measured with the adjusted mutual information (AMI), and the predictive accuracy with cross-validation. Since the map equation and the degree-corrected stochastic block model use stochastic search algorithms to detect communities, we average the results over ten searches for each of the 100 network samplings.

We analyze the Bayesian approach for prior $a \sim \ln(V)$. For the node degree, therefore, we use $a_{\alpha} = C \ln(V)$, where $\alpha=1 \dots V$ and $C$ is a constant that we need to specify. For the module degree, we use $a_{i\exit} = a_{i\enter} = \nu_i C \ln(V)$, where $\nu_i = V_i\frac{V-V_i}{V-1}$ for $i=1 \dots M$ and $V_i$ is the number of nodes in module $i$.

\subsection{Number of communities}

\begin{figure}
 \includegraphics[width=8.6cm]{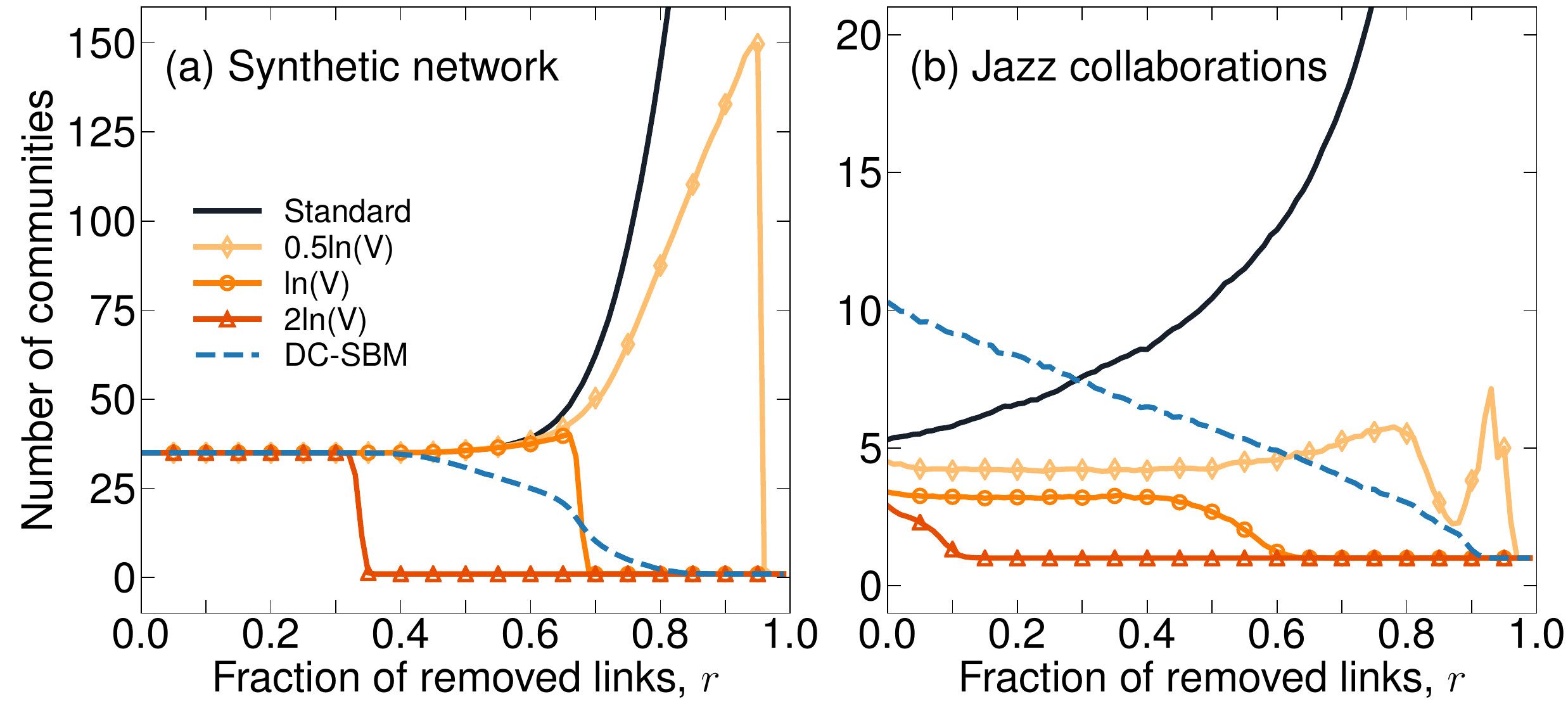}
 \caption{Mean number of communities obtained by the standard map equation, the Bayesian estimate of the map equation with different values of Dirichlet prior parameter $a$, and the degree-corrected stochastic block model (DC-SBM). The Bayesian estimate of the map equation with prior $a_{\alpha}=\ln(V)$ provides the best solution: when sufficient network data are available it distinguishes significant communities from mere noise, while in the undersampled regime it detects no community structure. Results are averaged over 100 network samplings and ten algorithm searches. The standard error of the mean is never higher than 0.58.}
 \label{fig:prior_nummodules}
\end{figure}

Applied to the synthetic network, the standard map equation favors the planted partition until we remove more than approximately $55\%$ of the links (Fig.~\ref{fig:prior_nummodules}(a)). As we remove more links, the network also becomes sparse within communities. In the undersampled regime below the detectability limit where it is not possible to recover the planted partition, the map equation overfits to random fluctuations and favors more, smaller communities. The Bayesian estimate of the map equation behaves differently. For $C=0.5$, the random prior network is weakly connected and cannot prevent overfitting when we remove $70-95\%$ of the links. In contrast, for $C=2$, the random prior network is densely connected and hides the communities in the noise induced by the prior such that the Bayesian estimate of the map equation underfits even when sufficient network data are available. In between, at the critical point where the random prior network becomes connected, the prior constant $C=1$ balances over- and underfitting and prevents the detection of spurious communities. Moreover, the amount of  noise that this prior network induces in the original network is so low that it does not wash out any significant community structure. While prior parameter $C$ between 0.5 and 1 performs best for some analyzed networks, $C=1$ remains a robust choice in general (Appendix~\ref{appendixC}).

The degree-corrected stochastic block model detects the planted partition until we remove more than $40\%$ of the links from the synthetic network. Compared to the Bayesian estimate of the map equation with the prior constant $C=1$, the degree-corrected stochastic block model starts to underfit the planted partition earlier. For $r>40\%$, the number of communities decreases continuously and when $r>80\%$, the degree-corrected stochastic block model detects no community structure.

Similar behaviors appear accentuated when we apply the methods to the real-world Jazz collaboration network (Fig.~\ref{fig:prior_nummodules}(b)). For the standard map equation, the number of detected communities increases with the number of missing links, whereas the degree-corrected stochastic block model shows the opposite trend. Unlike when applied to the synthetic network, the various map equation variants already favor different partitions before removing any links. The Bayesian estimate of the map equation detects fewer communities than the standard map equation, and its performance depends on the choice of the prior. For $C = 0.5$, the average number of communities is relatively stable when more than $50\%$ of the links remain. However, if we remove more than $50\%$ of the links, the number of communities increases because the prior parameter is too low. As for the synthetic network, the prior parameter $C = 2$ is too high and causes underfit: the method detects no community structure when we remove more than $10\%$ of the links. Again, $C = 1$ offers a good tradeoff. The number of communities is approximately constant as long as at least $50\%$ of the links remain and then decreases to 1 when fewer than $40\%$ of the links remain, where the method deduces that there no longer exists any significant community structure. 

\begin{figure}
 \includegraphics[width=8.6cm]{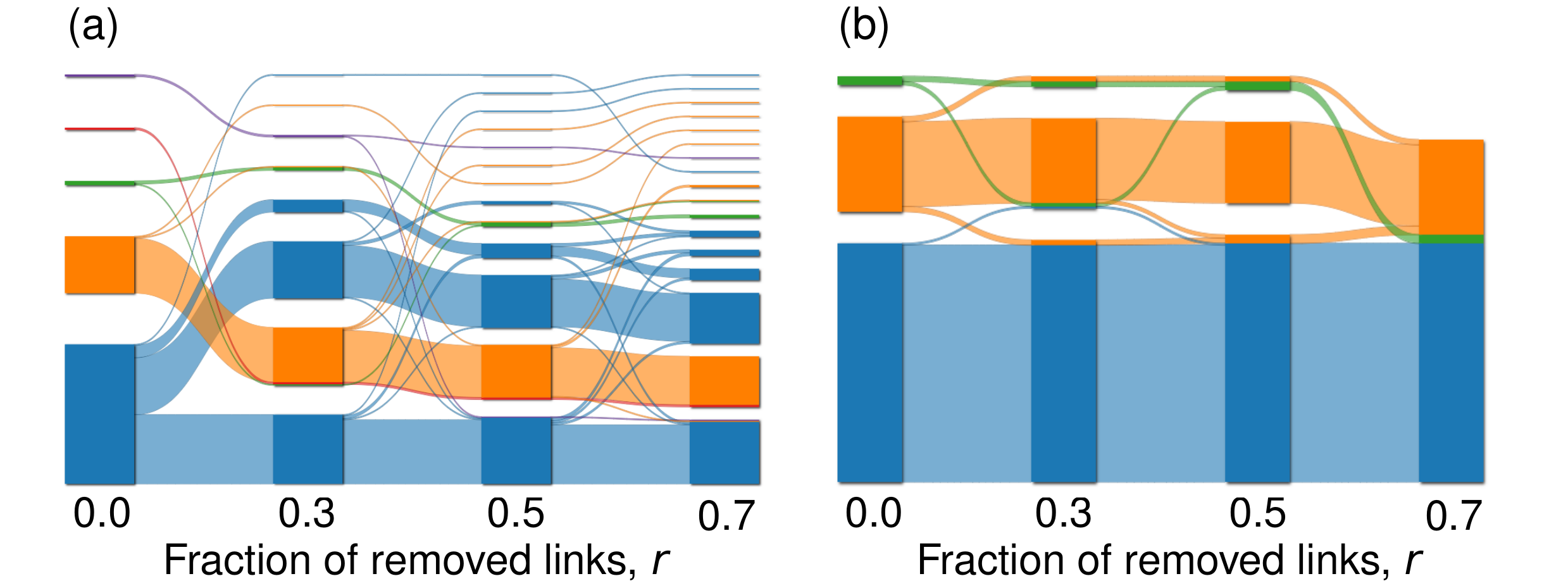}
 \caption{Alluvial diagrams of the Jazz collaboration network show changes in community structure with missing links for (a) the standard map equation and (b) the Bayesian estimate of the map equation with prior $a_{\alpha} = \ln(V)$. Compared to the standard map equation, the communities detected using the Bayesian estimate of the map equation are more robust to missing links.}
 \label{fig:alluvialDiagram}
\end{figure}

We illustrate differences in the community structure of the Jazz collaboration network induced by missing links for the standard and Bayesian map equation with alluvial diagrams~\cite{rosvall2010bootstrap}. The standard map equation identifies more and smaller communities with sparser networks, whereas its Bayesian estimate keeps similar communities with few changes before collapsing into one community when only 30\% of the links remain. The Bayesian estimate's prior assumption of missing links prevents the map equation from splitting communities when the networks lose links (Fig.~\ref{fig:alluvialDiagram}).

\subsection{Adjusted mutual information}

\begin{figure}[ht]
 \includegraphics[width=8.6cm]{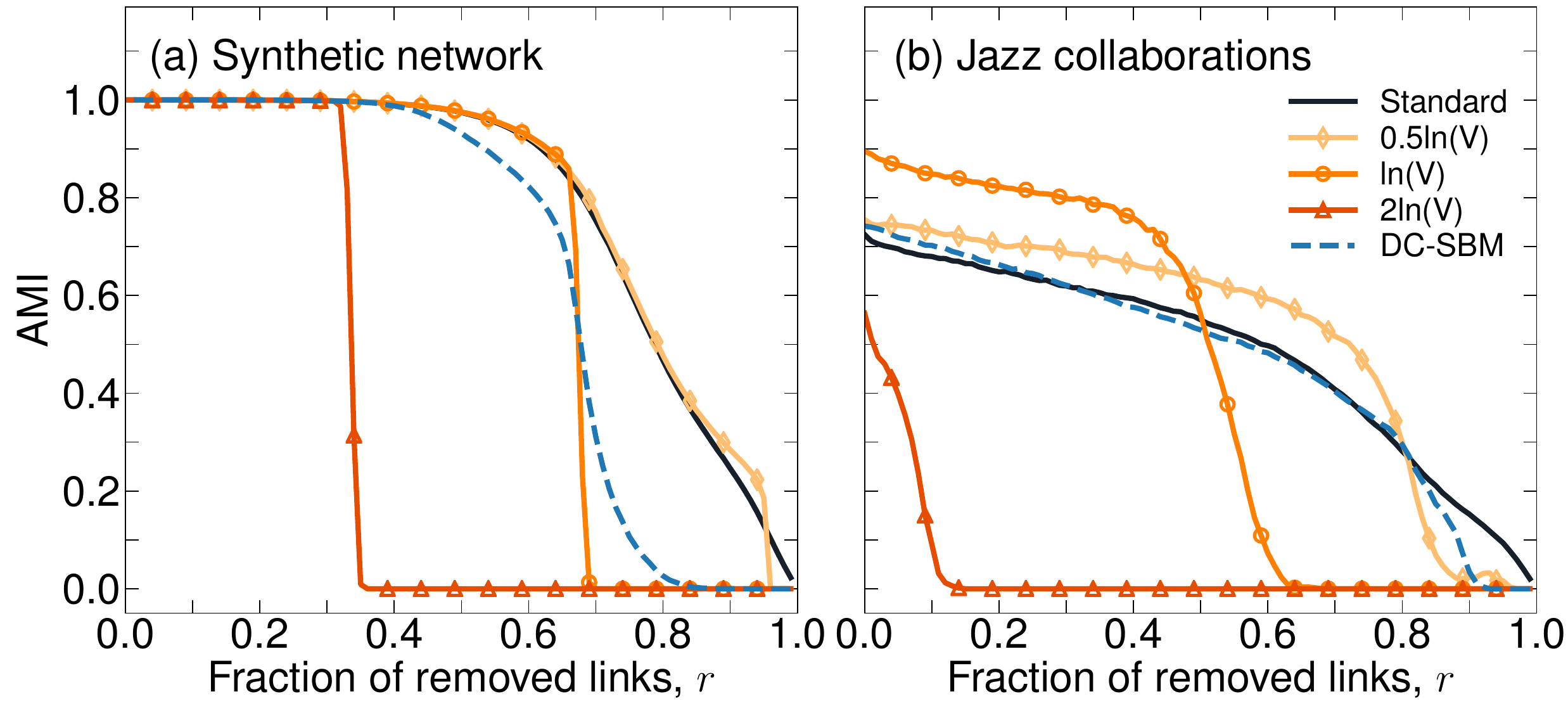}
 \caption{Performance tests of the community-detection algorithms using AMI. (a) AMI scores with the planted partition of the synthetic network as reference. (b) AMI scores with a partition obtained for the complete Jazz collaboration network as reference. The Bayesian estimate of the map equation with prior $a_{\alpha} = \ln(V)$ gives the most robust results when it is possible to detect significant communities. Results are averaged over 100 network samplings and ten algorithm searches. The standard error of the mean is never higher than 0.01.}
 \label{fig:ami}
\end{figure}

Adjusted mutual information (AMI) is a standard measure used to compare two different partitions \cite{vinh2010ami}. For the synthetic network, we compare identified partitions with the planted partition. The standard map equation successfully recovers the planted partition when more than $60\%$ of the links are available (AMI$=1$). When we remove more links, the accuracy decreases (Fig.~\ref{fig:ami}(a)). The Bayesian estimate of the map equation with prior constant $C=0.5$ has almost the same accuracy. If we use $C=1$ instead, the method performs slightly better when we remove $40-60\%$ of the links. Again, when we remove more than $65\%$ of the links, the Bayesian estimate of the map equation with prior constant $C=1$ deduces that there no longer exists any significant community structure and AMI$=0$.

To measure the AMI for the Jazz collaboration network, which has no planted partition, we compare the partitions that the community detection methods return for networks with different fractions of missing links to the partitions they return for the complete network. For the complete network, we measure the average AMI over ten searches. The Bayesian estimate of the map equation with prior $a_{\alpha} = \ln(V)$ is the most consistent method when it is possible to detect significant communities (Figure~\ref{fig:ami}(b)).

\begin{figure}[ht]
 \includegraphics[width=8.6cm]{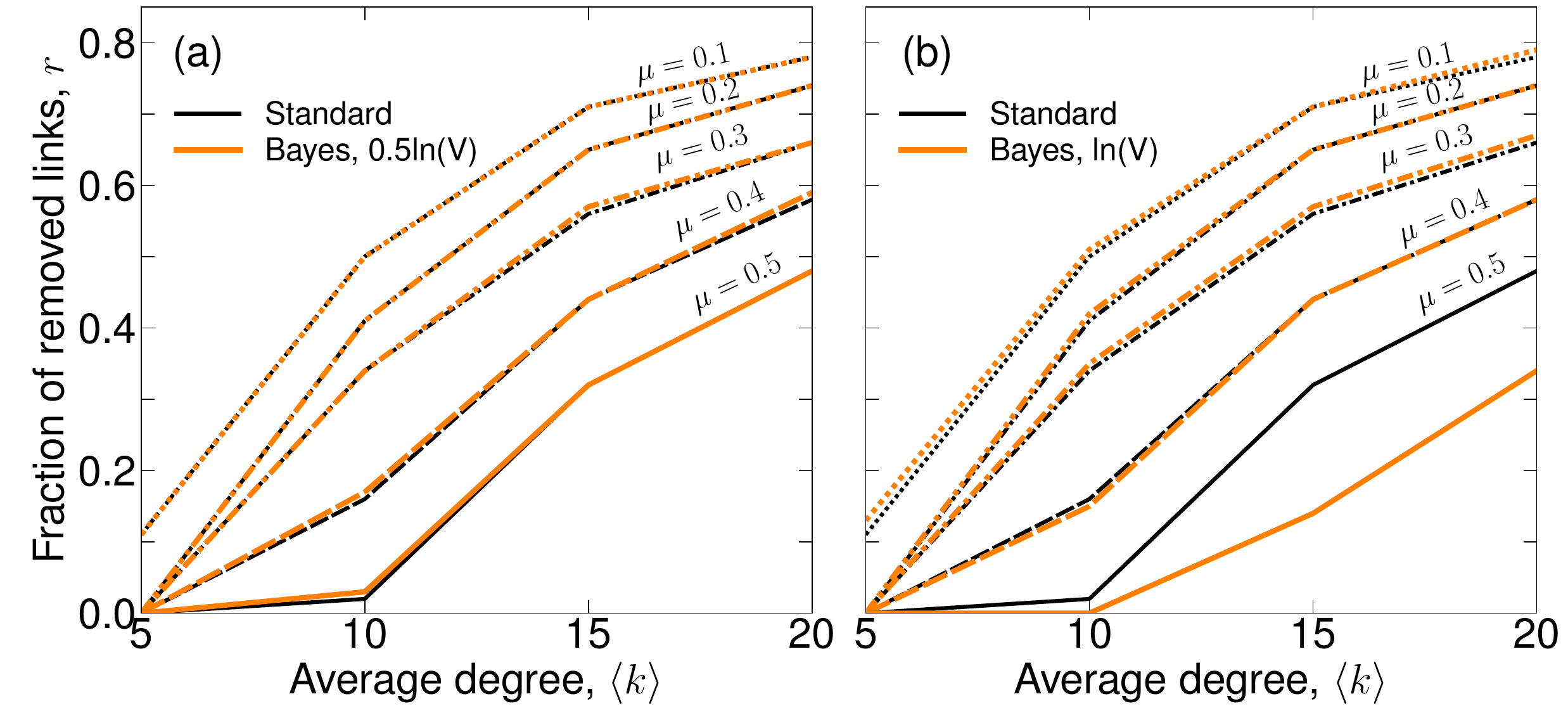}
 \caption{Impact of network structure on the performance of the standard map equation and its Bayesian estimate. Prior parameter $C = 0.5$ in (a) and $C = 1$ in (b). For LFR networks with $V=1000$ nodes and various densities $\langle k \rangle$ and mixing parameters $\mu$, we show the critical fraction of removed links $r(\langle k \rangle,\mu)$ where the AMI between the planted partition and the identified partition falls below 0.9. Except for weak community structures ($\mu=0.5$), where the Bayesian estimate with prior constant $C=1$ underfits for lower fraction of removed links than the standard map equation overfits, the methods are on par. Results are averaged over ten network samplings and ten algorithm searches.}
 \label{fig:lfrs}
\end{figure}

In both synthetic and real-world networks when $\langle k \rangle > \ln(V)$, the Bayesian estimate of the map equation with prior constant $C=1$ shows robust performance. However, when $C=2$ it can fail to detect their community structure due to the high level of noise induced by the prior. To understand how the noise induced by the prior in the Bayesian estimate of the map equation affects community detection in sparse networks with $\langle k \rangle \sim \ln(V)$ and weak community structure, we test the performance on a range of different networks. We generate LFR networks with various values of average degree and mixing parameter, randomly remove a fraction of links, detect communities using the standard map equation and its Bayesian estimate with prior $a_{\alpha} = 0.5\ln(V)$ and $\ln(V)$, and classify the community detection as successful when the AMI between the planted partition and the identified partition is 0.9 or higher. Even if the random prior network has higher density than the original network, the Bayesian estimate of the map equation achieves the same performance as the standard map equation when the community structure is well defined ($\mu<0.5$). However, if the community structure is weak ($\mu=0.5$), the prior $a_{\alpha} = \ln(V)$ can cause underfit before the standard map equation starts to overfit to noise induced by missing links (Fig.~\ref{fig:lfrs}). These results rely on the cost of overfitting and underfitting implied by the AMI. Specific networks or research questions may require other penalties for many or few communities.

\subsection{Cross-validation}

Cross-validation allows us to compare model-selection performance without planted or known partitions. 
We validate the significance of network partitions returned by Infomap for training networks with a fraction $1-r$ of links using the standard map equation and its Bayesian estimate (Fig.~\ref{fig:relativesavings}).

\begin{figure}
 \includegraphics[width=8.6cm]{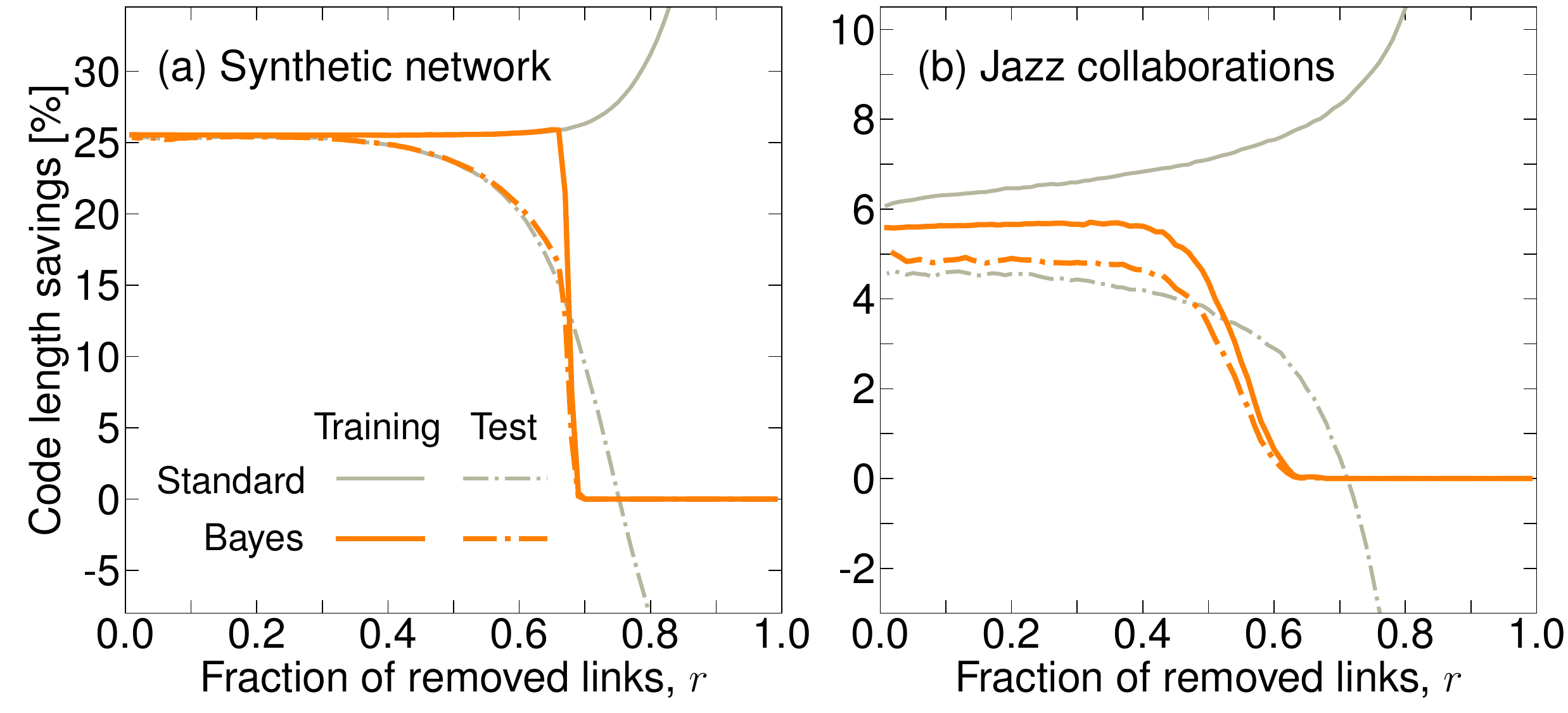}
 \caption{Performance tests of the map equation with and without Bayesian estimates using cross-validation. The Bayesian estimate of the map equation with prior $a_{\alpha} = \ln(V)$ prevents overfitting in the undersampled regime. Results are averaged over 100 network samplings and ten algorithm searches. The code length is measured with Grassberger entropy estimation. The standard error of the mean is never higher than 0.38.}
 \label{fig:relativesavings}
\end{figure}

As the link density of the training network decreases below the detectability limit, the standard map equation mistakes noisy substructures in the sparse training networks for actual communities. As a result, the relative code length savings in the training and test networks diverge, and partitions obtained with the standard map equation give negative code length savings in the test network. In contrast, the Bayesian estimate of the map equation with prior constant $C=1$ prevents overfitting in the sparse training network, implying that there is no significant community structure.

\begin{table*}
\caption{Comparison between partitions detected by the standard map equation and the Bayesian estimate of the map equation for six real-world networks. The notations m$_{0.25}$ and m$_{1.0}$ refer to the number of communities in the network with $25\%$ removed links and the complete network, respectively. The last two columns report the code length savings of test and training networks for partitions detected in the training networks with $25\%$ removed links}
\label{table:realnet}
\begin{tabular}
{Wl{1.5cm}Wr{1.5cm}Wr{1.7cm}Wc{0.8cm}Wl{2.3cm}Wr{0.5cm}Wr{1.2cm}Wc{0.4cm}Wr{1cm}Wr{1.1cm}}
\mytoprule
 Network & Nodes & Links & & Method & m$_{0}$ & m$_{0.25}$ & & $l_{0.25}^{train} (\%)$ & $l_{0.25}^{test} (\%)$ \\
 \mymidrule
 \rule{0pt}{3ex}\multirow{2}{1.6cm}{AstroPh} & \multirow{2}{*}{17,903} & \multirow{2}{*}{197,031} & & Bayes & 707 & 771 & & 24 & 18 \\
  & & & & Standard & 663 & 1,080 & & 24 & 18\\
\addlinespace
 \multirow{2}{1.6cm}{Email} & \multirow{2}{*}{1,133} & \multirow{2}{*}{5,451} & & Bayes & 34 & 1 & & 0 & 0 \\
  & & & & Standard & 50 & 104 & & 16 & 2 \\
\addlinespace
 \multirow{2}{1.6cm}{Erd\H{o}s N1} & \multirow{2}{*}{466} & \multirow{2}{*}{1,600} & & Bayes & 1 & 1 & & 0 & 0 \\
  & & & & Standard & 38 & 67 & & 17 & -9 \\
\addlinespace
 \multirow{2}{1.6cm}{Football} & \multirow{2}{*}{115} & \multirow{2}{*}{613} & & Bayes & 9 & 9 & & 18 & 15 \\
  & & & & Standard & 10 & 11 & & 20 & 16 \\
\addlinespace
 \multirow{2}{1.6cm}{PGP} & \multirow{2}{*}{10,680} & \multirow{2}{*}{24,316} & & Bayes & 956 & 1,057 & & 49 & 19 \\
  & & & & Standard & 897 & 2,210 & & 49 & 16 \\
\addlinespace
 \multirow{2}{1.6cm}{Polblogs} & \multirow{2}{*}{1,222} & \multirow{2}{*}{16,717} & & Bayes & 24 & 23 & & 6 & 5 \\
  & & & & Standard & 33 & 80 & & 6 & 5 \\
 \mybottomrule
 \end{tabular}
\end{table*}

To complement with results for other networks, we provide summary statistics for six real-world networks often used to evaluate the performance of community detection algorithms  (Table~\ref{table:realnet}). The networks include a collaboration network in Astrophysics extracted from the arXiv (AstroPh)~\cite{leskovec2007arxiv}, the network of e-mails exchanged between members of the University Rovira i Virgili (Email)~\cite{ebel2002email}, a collaboration network of authors with Erd\H{o}s number 1 (Erd\H{o}s N1)~\cite{batagelj2000erdos}, the American College Football network (Football)~\cite{newman2004football}, the PGP social network of trust (PGP)~\cite{boguna2004pgp}, and the network of political weblogs (Polblogs)~\cite{adamic2005polblogs}. In all networks, the standard map equation returns partitions with a higher number of communities when links are missing. Except for the Football network, the number of detected communities increases by $60\%$ or more compared with the number of communities detected in the complete network. In contrast, except for the AstroPh and PGP networks, the Bayesian estimate of the map equation with prior constant $C=1$ identifies partitions with fewer communities. Nevertheless, the different community structures detected by the two methods result in similar relative code length savings in all networks but the Email and Erd\H{o}s N1 networks. They are sparse with $\langle k \rangle < \ln(V)$. In the complete Email network, the Bayesian estimate of the map equation detects 34 communities but underfits and detects no community structure after removing $25\%$ of the links. After removing links in the Erd\H{o}s N1 network, the standard map equation overfits and detects communities that, when applied to the test network, gives worse compression than the one-module solution. The Bayesian estimate of the map equation prevents this overfitting by preferring the one-module solution over any non-trivial solution.

Overall, the model-accuracy results quantified by number of communities, AMI scores, and code length savings in cross-validation on synthetic and real-world networks suggest that the analyzed network and research question should determine whether to use the standard map equation or its Bayesian estimate. Choose the standard map equation when the network data are complete or when extra communities caused by missing links are not a problem. Choose its Bayesian estimate when spurious communities can harm the analysis.

\section{Conclusion}

We have derived a Bayesian approach of the map equation that imposes prior information about the network structure to reduce overfitting for sparse networks with missing links. Using an uninformative Dirichlet prior, we show that the Bayesian estimate of the map equation avoids finding spurious communities in sparse synthetic and real-world networks with missing links. With a properly chosen prior constant, the proposed method successfully balances the impact of the imposed prior against the observed network data: The Bayesian estimate of the map equation provides a principled approach to reducing overfitting and detecting significant communities in two or more levels. We also show how to asses whether communities are significant using more effective cross-validation with Grassberger entropy estimation, which enables larger training networks. The computational overhead of the methods compared with the standard map equation is low. We anticipate that more reliable flow-based community detection of undersampled networks will be useful in many applications, including better prediction of missing links.

\section*{ACKNOWLEDGMENTS}

    We thank Vincenzo Nicosia and Leto Peel for stimulating discussions and Christopher Bl\"{o}cker, Anton Eriksson, and Alexis Rojas for useful comments on the manuscript. J.S., D.E., and M.R.\ were supported by the Swedish Research Council, Grant No.\ 2016-00796.

\appendix

\section{Normalized transition rates}\label{appendixA}

\begin{proposition}{}
The map equation, 
\begin{align}
\begin{split}
 L&(\mathsf{M}) = -\sum_{\alpha=1}^{V}p_{\alpha}\log_2(p_{\alpha})-\sum_{i=1}^{m}q_{i\curvearrowright}\log_2(q_{i\curvearrowright})\\
       & +\sum_{i=1}^{m}\left(q_{i\curvearrowright}+\sum_{\alpha \in i} p_{\alpha}\right)\log_2\left(q_{i\curvearrowright}+\sum_{\alpha \in i} p_{\alpha}\right)\\
       & -\sum_{i=1}^{m}q_{i\curvearrowleft}\log_2(q_{i\curvearrowleft}) + \left(\sum_{i=1}^{m}q_{i\curvearrowleft}\right)\log_2\left(\sum_{i=1}^{m}q_{i\curvearrowleft}\right),
\end{split}
\end{align}
is a scale invariant function.
\end{proposition}
 
\begin{proof}
If we scale the transition rates $p_{\alpha}, q_{i\curvearrowleft}$ and $q_{i\curvearrowright}$ by a constant $K$, where $K >0$, and change $L(\mathsf{M})$ to 
\begin{align*}
\begin{split}
 &L'_{\mathsf{M}} = -\sum_{\alpha=1}^{V}Kp_{\alpha}\log_2(Kp_{\alpha})-\sum_{i=1}^{m}Kq_{i\curvearrowright}\log_2(Kq_{i\curvearrowright})\\
       & +\sum_{i=1}^{m}\left(Kq_{i\curvearrowright}+\sum_{\alpha \in i} Kp_{\alpha}\right)\log_2\left(Kq_{i\curvearrowright}+\sum_{\alpha \in i} Kp_{\alpha}\right)\\
       & -\sum_{i=1}^{m}Kq_{i\curvearrowleft}\log_2(Kq_{i\curvearrowleft}) + \left(\sum_{i=1}^{m}Kq_{i\curvearrowleft}\right)\log_2\left(\sum_{i=1}^{m}Kq_{i\curvearrowleft}\right),
\end{split}
\end{align*}  
then
\begin{align*}
\begin{split}
&L'_{\mathsf{M}} = -{\cancel{K\sum_{\alpha=1}^{V}p_{\alpha}\log_2(K)}} 
        -K\sum_{\alpha=1}^{V}p_{\alpha}\log_2(p_{\alpha}) \\
      & -{\cancel{K\sum_{i=1}^{m}q_{i\curvearrowright}\log_2(K)}}-K\sum_{i=1}^{m}q_{i\curvearrowright}\log_2(q_{i\curvearrowright})\\
      & -{\cancel{K\sum_{i=1}^{m}q_{i\curvearrowleft}\log_2(K)}}
       -K\sum_{i=1}^{m}q_{i\curvearrowleft}\log_2(q_{i\curvearrowleft})\\
      &+{\cancel{K\sum_{i=1}^{m}q_{i\curvearrowright}\log_2(K)}} + {\cancel{K\sum_{\alpha=1}^{V}p_{\alpha}\log_2(K)}}\\
      &+K\sum_{i=1}^{m}\left(q_{i\curvearrowright}+\sum_{\alpha \in i} p_{\alpha}\right)\log_2\left(q_{i\curvearrowright}
       +\sum_{\alpha \in i} p_{\alpha}\right) \\
      &+ {\cancel{K\left(\sum_{i=1}^{m}q_{i\curvearrowleft}\right)\log_2(K)}}
       + K\left(\sum_{i=1}^{m}q_{i\curvearrowleft}\right)\log_2\left(\sum_{i=1}^{m}q_{i\curvearrowleft}\right)\\
    =& KL(\mathsf{M})    
\end{split} 
\end{align*}
\end{proof}

If we choose 
\begin{align}
K = \frac{\sum_{\alpha=1}^{V}k_{\alpha}}{\sum_{\alpha=1}^{V}k_{\alpha} + \sum_{i=1}^{m}k_{i\curvearrowleft} + \sum_{i=1}^{m}k_{i\curvearrowright}} 
\end{align}
such that
\begin{align}
 p'_{\alpha} = K p_{\alpha}=\frac{k_{\alpha}}{\sum_{\alpha=1}^{V}k_{\alpha} + \sum_{i=1}^{m}k_{i\curvearrowleft} + \sum_{i=1}^{m}k_{i\curvearrowright}}
\end{align}
\begin{align}
 q'_{i\curvearrowleft} = K q_{i\curvearrowleft} = \frac{k_{i\curvearrowleft}}{\sum_{\alpha=1}^{V}k_{\alpha} + \sum_{i=1}^{m}k_{i\curvearrowleft} + \sum_{i=1}^{m}k_{i\curvearrowright}}
\end{align}
\begin{align}
 q'_{i\curvearrowright} = K q_{i\curvearrowright} = \frac{k_{i\curvearrowright}}{\sum_{\alpha=1}^{V}k_{\alpha} + \sum_{i=1}^{m}k_{i\curvearrowleft} + \sum_{i=1}^{m}k_{i\curvearrowright}}
\end{align}
we will have
\begin{align}
 \sum_{\alpha=1}^{V}p'_{\alpha} + \sum_{i=1}^{m}q'_{i\curvearrowleft} + \sum_{i=1}^{m}q'_{i\curvearrowright} = 1.
\end{align}

Now we can use
\begin{align}
\begin{split}
 L&(\mathsf{M}) =\frac{1}{K}[ -\sum_{\alpha=1}^{V}p'_{\alpha}\log_2(p'_{\alpha})-\sum_{i=1}^{m}q'_{i\curvearrowright}\log_2(q'_{i\curvearrowright})\\
 &+\sum_{i=1}^{m}\left(q'_{i\curvearrowright}+\sum_{\alpha \in i} p'_{\alpha}\right)\log_2\left(q'_{i\curvearrowright}+\sum_{\alpha \in i} p'_{\alpha}\right)\\ 
 &-\sum_{i=1}^{m}q'_{i\curvearrowleft}\log_2(q'_{i\curvearrowleft})
        + \left(\sum_{i=1}^{m}q'_{i\curvearrowleft}\right)\log_2\left(\sum_{i=1}^{m}q'_{i\curvearrowleft}\right)]
\end{split}
\end{align}
to calculate the posterior average of the map equation
\begin{align}
\begin{split}
 \hat{L}_B(\mathsf{M}) =& E[L(\mathsf{M})|\pmb{k},\pmb{a}] \\
                =& \int L(\mathsf{M}) P(\pmb{p}',\pmb{q}'_{\curvearrowleft},\pmb{q}'_{\curvearrowright}|\pmb{k},\pmb{a})d\pmb{p}'d\pmb{q}'_{\curvearrowleft}d\pmb{q}'_{\curvearrowright}
\end{split}                
\end{align}
where posterior probability distribution equals
\begin{align}
\begin{split}
 &P(\pmb{p}',\pmb{q}'_{\curvearrowleft},\pmb{q}'_{\curvearrowright}|\pmb{k},\pmb{a}) \propto \\
 &\prod_{\alpha=1}^{V} (p'_{\alpha})^{k_{\alpha}+a_{\alpha}-1} \prod_{i=1}^{m} \left[(q'_{i\curvearrowright})^{k_{i\curvearrowright}+a_{i\curvearrowright}-1} (q'_{i\curvearrowleft})^{k_{i\curvearrowleft}+a_{i\curvearrowleft}-1}\right].
\end{split} 
\end{align}
As a result we obtain
\begin{align}
 \begin{split}
  \hat{L}_B&(\mathsf{M})= \frac{1}{\ln(2)} \frac{1}{\sum_{\alpha=1}^{V}u_{\alpha}} \\
            &\times [ -\sum_{\alpha=1}^{V}u_{\alpha}\psi(u_{\alpha}+1) - \sum_{i=1}^{m}u_{i\curvearrowright}\psi(u_{i\curvearrowright}+1) \\
            & + \sum_{i=1}^{m} ( u_{i\curvearrowright} + \sum_{\alpha \in i} u_{\alpha} ) \psi ( u_{i\curvearrowright} + \sum_{\alpha \in i} u_{\alpha} + 1 )\\
            &- \sum_{i=1}^{m}u_{i\curvearrowleft}\psi(u_{i\curvearrowleft}+1)                 
              + (\sum_{i=1}^{m}u_{i\curvearrowleft})\psi(\sum_{i=1}^{m}u_{i\curvearrowleft} + 1) ]
 \end{split}
\end{align}
where $u_x=k_x+a_x$ and $\psi$ is digamma function, $\psi(x)=\frac{d}{dx}\ln(\Gamma(x))$.

\section{The Bayesian estimate of the multilevel map equation}\label{appendixB}

The multilevel formulation of the map equation~\cite{rosvall2011multilevelmapeq, edler2017mapequation} measures the minimum average description length given a multilevel map $\mathsf{M}$ of $V$ nodes clustered into $m$ communities, for which each community $i$ has a submap $\mathsf{M}_i$ with $m_i$ subcommunities, for which each subcommunity $ij$ has a submap $\mathsf{M}_{ij}$ with $m_{ij}$ subcommunities, and~so on. It uses hierarchically nested code structures,
\begin{align}
L(\mathsf{M}) = q_{\curvearrowleft} H(\mathcal{Q}) + \sum_{i=1}^{m}L(\mathsf{M}_i), \label{eq:requirsivemapeq}
\end{align}
where the average per-step code length needed to describe random walker transitions between communities at the coarsest level is the same as in the case of two-level clusterings,
\begin{align}
H(\mathcal{Q}) = -\sum_{i=1}^{m}\frac{q_{i\curvearrowleft}}{q_{\curvearrowleft}} \log_2{\frac{q_{i\curvearrowleft}}{q_{\curvearrowleft}}},
\end{align}
and the average per-step code word length of the module codebook $i$ recursively takes into account contributions of the description lengths of communities at finer levels,
\begin{align}
L(\mathsf{M}_i) = q^{\circlearrowright}_i H(\mathcal{Q}_i) + \sum_{j=1}^{m_i}L(\mathsf{M}_{ij}). \label{eq:Lintermediate}
\end{align}
Here, the average per-step code length needed to describe the random walker at intermediate level $i$ exiting to a coarser level or entering the $m_{i}$ subcommunities $\mathsf{M}_{ij}$ at a finer level is
\begin{align}
H(\mathcal{Q}_{i}) = -\frac{q_{i\curvearrowright}}{q^{\circlearrowright}_{i}} \log_2{\frac{q_{i\curvearrowright}}{q^{\circlearrowright}_{i}}} -\sum_{j=1}^{m_i}\frac{q_{ij\curvearrowleft}}{q^{\circlearrowright}_{i}} \log_2{\frac{q_{ij\curvearrowleft}}{q^{\circlearrowright}_{i}}},
\end{align}
where
\begin{align}
q^{\circlearrowright}_{i} = q_{i\curvearrowright} + \sum_{j=1}^{m_i}q_{ij\curvearrowleft}
\end{align} 
is the total code rate use in subcommunity $i$.
We add the description lengths of codebooks for subcommunities at finer levels in a recursive fashion down to the finest level,
\begin{equation}
L(\mathsf{M}_{ij\ldots l}) = p^{\circlearrowright}_{ij\ldots l}H(\mathcal{P}_{ij\ldots l}),\label{eq:Lfinest}
\end{equation}
where
\begin{align}
\begin{split}
H(\mathcal{P}_{ij\ldots l}) = &-\frac{q_{ij\ldots l\curvearrowright}}{p^{\circlearrowright}_{ij\ldots l}}\log_2{\frac{q_{ij\ldots l\curvearrowright}}{p^{\circlearrowright}_{ij\ldots l}}} \\
                            &- \sum_{\alpha \in \mathsf{M}_{ij\ldots l}}\frac{\pi_\alpha}{p^{\circlearrowright}_{ij\ldots l}}\log_2{\frac{\pi_\alpha}{p^{\circlearrowright}_{ij\ldots l}}}\label{eq:Hijl}
\end{split}
\end{align}
and
\begin{equation}
p^{\circlearrowright}_{ij\ldots l} = q_{ij\ldots l\curvearrowright} + \sum_{\alpha \in \mathsf{M}_{ij\ldots l}}\pi_\alpha \label{eq:usemno}
\end{equation}
is the total code word use rate of module codebook $ij\ldots l$.

To obtain the Bayesian estimate of the multilevel map equation, we use Eq.~(\ref{eq:requirsivemapeq}) to calculate the posterior average according to Eq.~(\ref{eq:estimatorBayes}). Following the same procedure described in Sec.~\ref{sec:bayesmapeq}, we obtain a formula for the Bayesian estimate of the multilevel map equation,
\begin{align}
\begin{split}
&\hat{L}_B(\mathsf{M})= \frac{1}{\ln(2)}\frac{1}{\sum_{\alpha=1}^{V}u_{\alpha}}\\
&\times[-\sum_{i=1}^{m}u_{i\curvearrowleft}\psi(u_{i\curvearrowleft}+1)+(\sum_{i=1}^{m}u_{i\curvearrowleft})\psi(\sum_{i=1}^{m}u_{i\curvearrowleft}+1)]\\
&+ \sum_{i=1}^{m}\hat{L}_B(\mathsf{M}_i),
\end{split}    
\end{align}
where
\begin{align}
\begin{split}
&\hat{L}_B(\mathsf{M}_i)= \frac{1}{\ln(2)}\frac{1}{\sum_{\alpha=1}^{V}u_{\alpha}}\\
&\times[-u_{i\curvearrowright}\psi(u_{i\curvearrowright}+1)-\sum_{j=1}^{m_i}u_{ij\curvearrowleft}\psi(u_{ij\curvearrowleft}+1)\\
&+(u_{i\curvearrowright}+\sum_{j=1}^{m_i}u_{ij\curvearrowleft})\psi(u_{i\curvearrowright}+\sum_{j=1}^{m_i}u_{ij\curvearrowleft}+1)]\\
&+ \sum_{j=1}^{m_i}\hat{L}_B(\mathsf{M}_{ij})
\end{split}    
\end{align}
and at the finest level
\begin{align}
\begin{split}
&\hat{L}_B(\mathsf{M}_{ij\dots l})=\frac{1}{\ln(2)}\frac{1}{\sum_{\alpha=1}^{V}u_{\alpha}}\\
&\times[-u_{ij\dots l \curvearrowright}\psi(u_{ij\dots l\curvearrowright} + 1) 
- \sum_{\alpha \in \mathsf{M}_{ij\dots l}}u_{\alpha}\psi(u_{\alpha}+1)\\
& + (u_{ij\dots l\curvearrowright}+\sum_{\alpha \in \mathsf{M}_{ij\dots l}}u_{\alpha})\psi(u_{ij\dots l\curvearrowright}+\sum_{\alpha \in \mathsf{M}_{ij\dots l}}u_{\alpha} + 1)].
\end{split}    
\end{align}

\section{Results for different values of the prior parameter}\label{appendixC}

\begin{figure}[ht]
 \includegraphics[width=8.6cm]{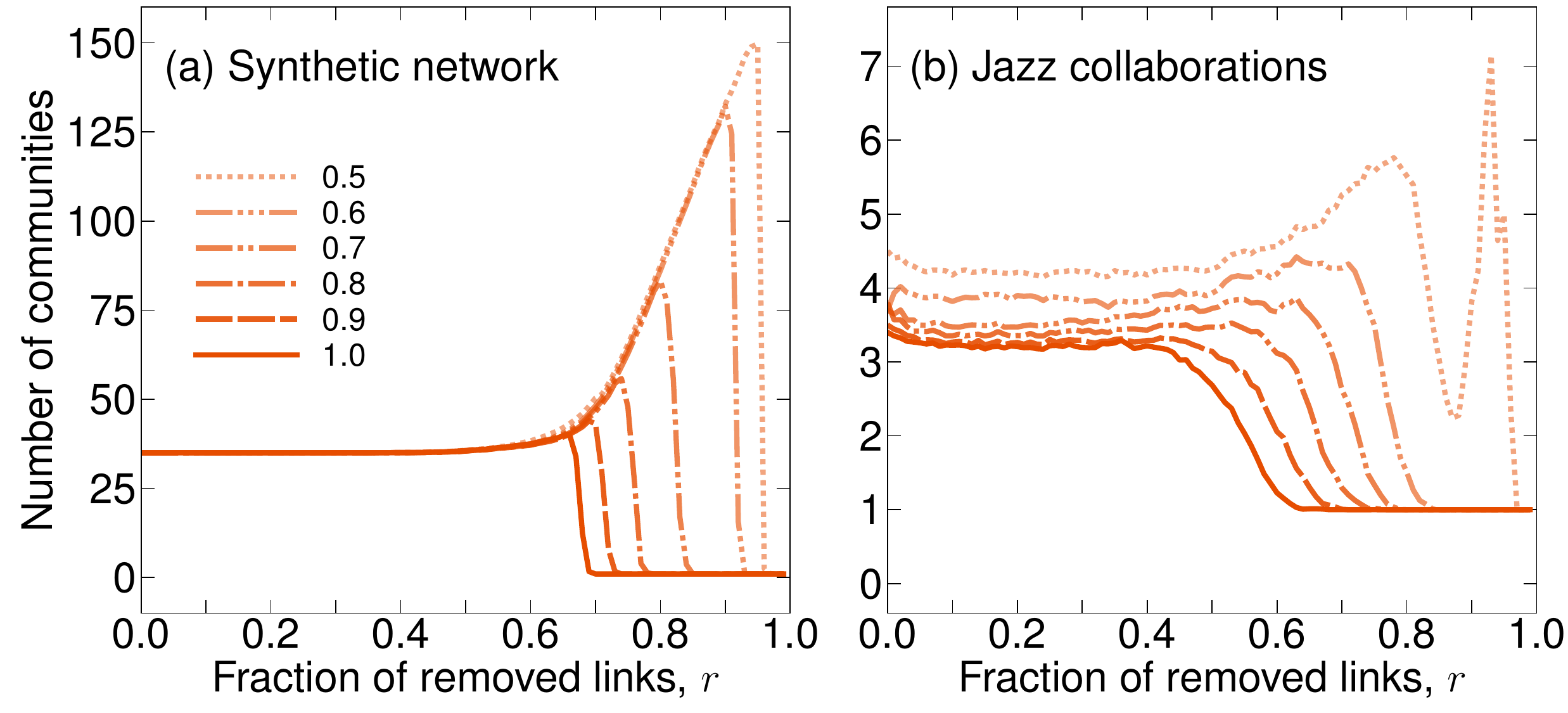}
 \caption{Mean number of communities obtained by the Bayesian estimate of the map equation with different values of the prior constant $C$. Smaller prior constants give more communities when many links are missing. Results are averaged over 100 network samplings and ten algorithm searches.}
 \label{fig:nummod_priors}
\end{figure}

The number of communities obtained by the Bayesian estimate of the map equation varies for different values of the prior constant $C$ between 0.5 and 1 (Fig.~\ref{fig:nummod_priors}). For the synthetic network in the undersampled regime, $C<0.8$ can lead to severe overfitting before removing so many links that it becomes evident that there is no significant community structure. For the Jazz collaboration network, the number of detected communities is similar for prior constant $C>0.6$ but is higher for all values of $r$ when $C \leq 0.6$.

\begin{figure}[ht]
 \includegraphics[width=8.6cm]{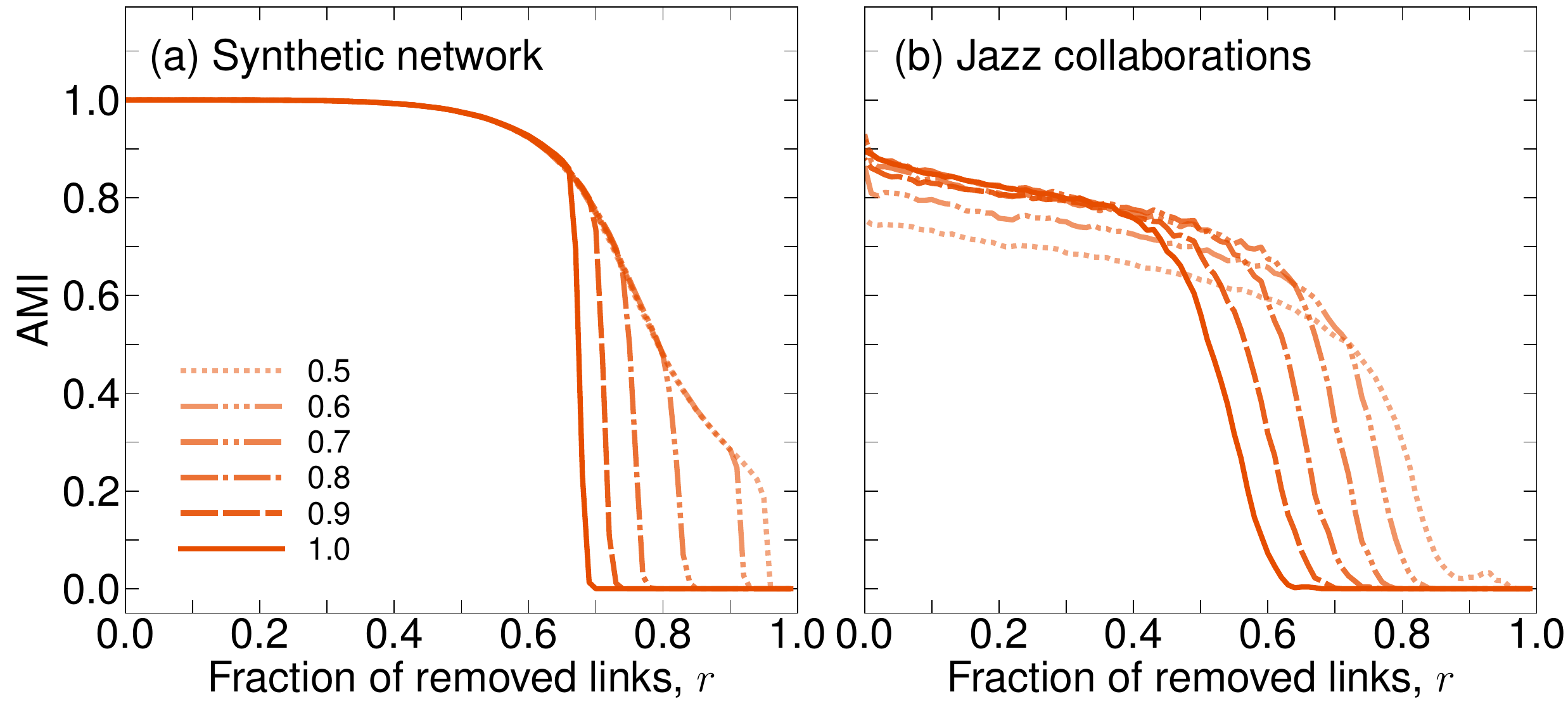}
 \caption{Performance tests of the Bayesian estimate of the map equation with different values of the prior constant $C$ using AMI. (a) AMI scores with the planted partition of the synthetic network as reference. (b) AMI scores with a partition obtained for the complete Jazz collaboration network as reference. Smaller prior constants give communities with non-zero AMI scores when many links are missing at the cost of overall lower AMI-scores in the Jazz network. Results are averaged over 100 network samplings and ten algorithm searches.}
 \label{fig:ami_priors}
\end{figure}

To compare the performance for different prior parameters, we also compute the AMI for $C$ between 0.5 and 1 (Fig.~\ref{fig:ami_priors}). For the synthetic network, the AMI results confirm that the detected communities become sensitive to the choice of prior when we remove more than $65\%$ of the links. For example, for $C \geq 0.8$, the detected communities have AMI down to $0.65$ before dropping to 0. For $C<0.8$, the method can detect communities in sparser networks but these communities have AMI scores below 0.5. For the Jazz collaboration network, the AMI results confirm that the detected communities are more robust when $C > 0.6$.

\begin{figure}[ht]
 \includegraphics[width=8.6cm]{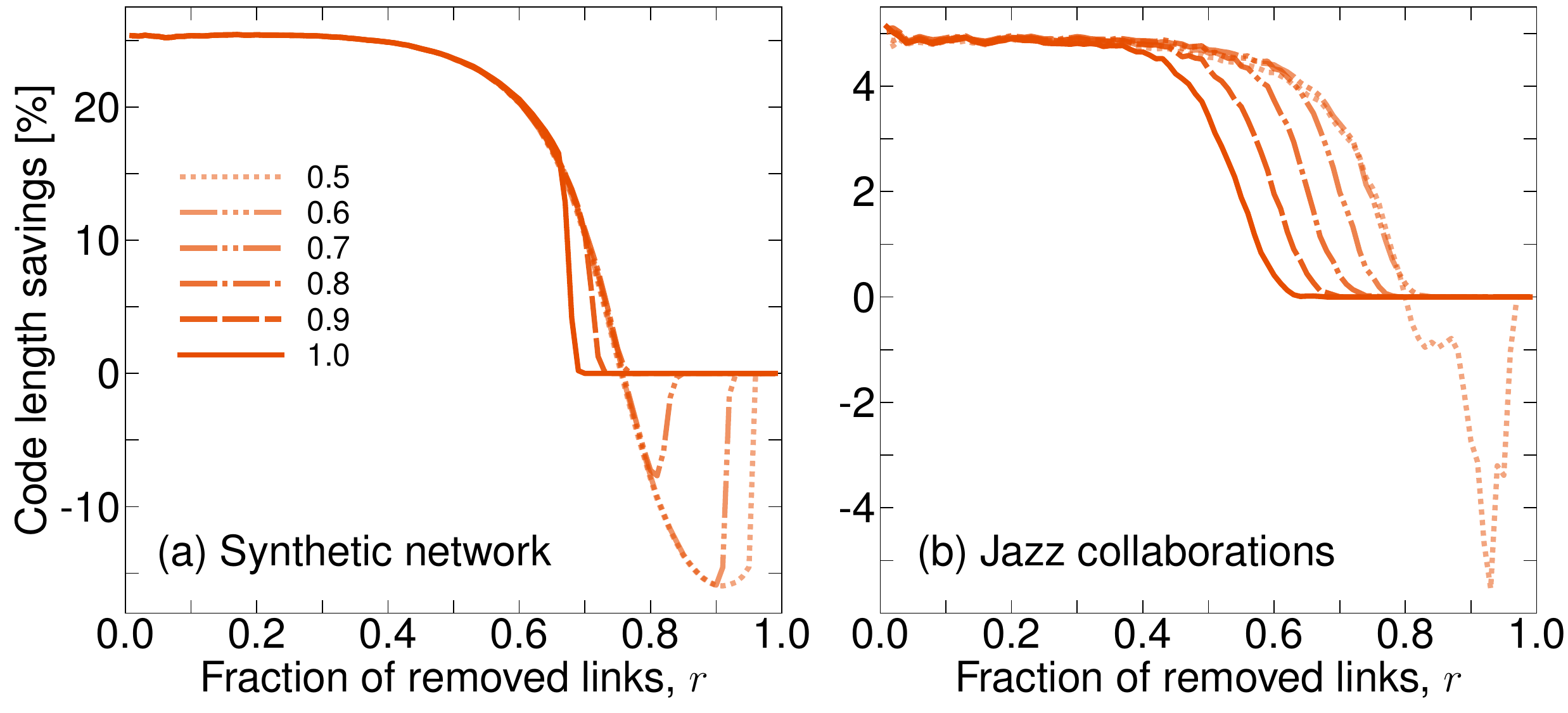}
 \caption{Performance tests of the Bayesian estimate of the map equation with different values of the prior constant $C$ using cross-validation.  Smaller prior constants give higher compression in a narrow range of missing links at the cost of lower compression for more missing links. We show relative code length savings for the test network compared to the one-community partition. The code length is measured with Grassberger entropy estimation. Results are averaged over 100 network samplings and ten algorithm searches.}
 \label{fig:crossvalidation_priors}
\end{figure}

Cross-validation further confirms these results for different prior parameters. For the synthetic network, the Bayesian estimate of the map equation is more robust to overfitting with prior constant $C \geq 0.8$ (Fig.~\ref{fig:crossvalidation_priors}). With $C<0.8$ and more than $75\%$ of the links removed, the communities detected in the training network applied to the test network give worse compression than with a single community. For the Jazz collaboration network, a prior with $C \geq 0.6$ prevents the detection of communities in the training network that, when applied to the test network, give negative relative code length savings.

These results for different values of the prior parameter indicate that there is no single prior $C\ln(V)$ that achieves optimal performance for all networks. We suggest using $\ln(V)$ as a prior because it is robust to overfitting and has good overall performance. If desired for specific networks, $C$ can be optimized between $0.5$ and 1 with cross-validation.

\newpage


\begin{thebibliography}{44}%
\makeatletter
\providecommand \@ifxundefined [1]{%
 \@ifx{#1\undefined}
}%
\providecommand \@ifnum [1]{%
 \ifnum #1\expandafter \@firstoftwo
 \else \expandafter \@secondoftwo
 \fi
}%
\providecommand \@ifx [1]{%
 \ifx #1\expandafter \@firstoftwo
 \else \expandafter \@secondoftwo
 \fi
}%
\providecommand \natexlab [1]{#1}%
\providecommand \enquote  [1]{``#1''}%
\providecommand \bibnamefont  [1]{#1}%
\providecommand \bibfnamefont [1]{#1}%
\providecommand \citenamefont [1]{#1}%
\providecommand \href@noop [0]{\@secondoftwo}%
\providecommand \href [0]{\begingroup \@sanitize@url \@href}%
\providecommand \@href[1]{\@@startlink{#1}\@@href}%
\providecommand \@@href[1]{\endgroup#1\@@endlink}%
\providecommand \@sanitize@url [0]{\catcode `\\12\catcode `\$12\catcode
  `\&12\catcode `\#12\catcode `\^12\catcode `\_12\catcode `\%12\relax}%
\providecommand \@@startlink[1]{}%
\providecommand \@@endlink[0]{}%
\providecommand \url  [0]{\begingroup\@sanitize@url \@url }%
\providecommand \@url [1]{\endgroup\@href {#1}{\urlprefix }}%
\providecommand \urlprefix  [0]{URL }%
\providecommand \Eprint [0]{\href }%
\providecommand \doibase [0]{https://doi.org/}%
\providecommand \selectlanguage [0]{\@gobble}%
\providecommand \bibinfo  [0]{\@secondoftwo}%
\providecommand \bibfield  [0]{\@secondoftwo}%
\providecommand \translation [1]{[#1]}%
\providecommand \BibitemOpen [0]{}%
\providecommand \bibitemStop [0]{}%
\providecommand \bibitemNoStop [0]{.\EOS\space}%
\providecommand \EOS [0]{\spacefactor3000\relax}%
\providecommand \BibitemShut  [1]{\csname bibitem#1\endcsname}%
\let\auto@bib@innerbib\@empty
\bibitem [{\citenamefont {Pons}\ and\ \citenamefont
  {Latapy}(2006)}]{pons2005flows}%
  \BibitemOpen
  \bibfield  {author} {\bibinfo {author} {\bibfnamefont {P.}~\bibnamefont
  {Pons}}\ and\ \bibinfo {author} {\bibfnamefont {M.}~\bibnamefont {Latapy}},\
  }\href@noop {} {\bibfield  {journal} {\bibinfo  {journal} {J. Graph
  Algorithms Appl.}\ }\textbf {\bibinfo {volume} {10}},\ \bibinfo {pages} {191}
  (\bibinfo {year} {2006})}\BibitemShut {NoStop}%
\bibitem [{\citenamefont {Rosvall}\ and\ \citenamefont
  {Bergstrom}(2008)}]{rosvall2008maps}%
  \BibitemOpen
  \bibfield  {author} {\bibinfo {author} {\bibfnamefont {M.}~\bibnamefont
  {Rosvall}}\ and\ \bibinfo {author} {\bibfnamefont {C.~T.}\ \bibnamefont
  {Bergstrom}},\ }\href@noop {} {\bibfield  {journal} {\bibinfo  {journal}
  {Proc. Natl. Acad. Sci. USA}\ }\textbf {\bibinfo {volume} {105}},\ \bibinfo
  {pages} {1118} (\bibinfo {year} {2008})}\BibitemShut {NoStop}%
\bibitem [{\citenamefont {Delvenne}\ \emph {et~al.}(2010)\citenamefont
  {Delvenne}, \citenamefont {Yaliraki},\ and\ \citenamefont
  {Barahona}}]{delvenne2010flows}%
  \BibitemOpen
  \bibfield  {author} {\bibinfo {author} {\bibfnamefont {J.-C.}\ \bibnamefont
  {Delvenne}}, \bibinfo {author} {\bibfnamefont {S.~N.}\ \bibnamefont
  {Yaliraki}},\ and\ \bibinfo {author} {\bibfnamefont {M.}~\bibnamefont
  {Barahona}},\ }\href@noop {} {\bibfield  {journal} {\bibinfo  {journal}
  {Proc. Natl. Acad. Sci. USA}\ }\textbf {\bibinfo {volume} {107}},\ \bibinfo
  {pages} {12755} (\bibinfo {year} {2010})}\BibitemShut {NoStop}%
\bibitem [{\citenamefont {Schaub}\ \emph {et~al.}(2012)\citenamefont {Schaub},
  \citenamefont {Delvenne}, \citenamefont {Yaliraki},\ and\ \citenamefont
  {Barahona}}]{schaub2012flows}%
  \BibitemOpen
  \bibfield  {author} {\bibinfo {author} {\bibfnamefont {M.~T.}\ \bibnamefont
  {Schaub}}, \bibinfo {author} {\bibfnamefont {J.-C.}\ \bibnamefont
  {Delvenne}}, \bibinfo {author} {\bibfnamefont {S.~N.}\ \bibnamefont
  {Yaliraki}},\ and\ \bibinfo {author} {\bibfnamefont {M.}~\bibnamefont
  {Barahona}},\ }\href@noop {} {\bibfield  {journal} {\bibinfo  {journal} {PLoS
  One}\ }\textbf {\bibinfo {volume} {7}},\ \bibinfo {pages} {e32210} (\bibinfo
  {year} {2012})}\BibitemShut {NoStop}%
\bibitem [{\citenamefont {Lambiotte}\ \emph {et~al.}(2014)\citenamefont
  {Lambiotte}, \citenamefont {Delvenne},\ and\ \citenamefont
  {Barahona}}]{lambiotte2014flows}%
  \BibitemOpen
  \bibfield  {author} {\bibinfo {author} {\bibfnamefont {R.}~\bibnamefont
  {Lambiotte}}, \bibinfo {author} {\bibfnamefont {J.}~\bibnamefont
  {Delvenne}},\ and\ \bibinfo {author} {\bibfnamefont {M.}~\bibnamefont
  {Barahona}},\ }\href@noop {} {\bibfield  {journal} {\bibinfo  {journal} {IEEE
  Trans. Network Sci. Eng.}\ }\textbf {\bibinfo {volume} {1}},\ \bibinfo
  {pages} {76} (\bibinfo {year} {2014})}\BibitemShut {NoStop}%
\bibitem [{\citenamefont {Edler}\ \emph {et~al.}(2017)\citenamefont {Edler},
  \citenamefont {Bohlin},\ and\ \citenamefont
  {Rosvall}}]{edler2017mapequation}%
  \BibitemOpen
  \bibfield  {author} {\bibinfo {author} {\bibfnamefont {D.}~\bibnamefont
  {Edler}}, \bibinfo {author} {\bibfnamefont {L.}~\bibnamefont {Bohlin}},\ and\
  \bibinfo {author} {\bibfnamefont {M.}~\bibnamefont {Rosvall}},\ }\href@noop
  {} {\bibfield  {journal} {\bibinfo  {journal} {Algorithms}\ }\textbf
  {\bibinfo {volume} {10}},\ \bibinfo {pages} {112} (\bibinfo {year}
  {2017})}\BibitemShut {NoStop}%
\bibitem [{\citenamefont {Guimer{\`a}}\ and\ \citenamefont
  {Sales-Pardo}(2009)}]{guimera2009prediction}%
  \BibitemOpen
  \bibfield  {author} {\bibinfo {author} {\bibfnamefont {R.}~\bibnamefont
  {Guimer{\`a}}}\ and\ \bibinfo {author} {\bibfnamefont {M.}~\bibnamefont
  {Sales-Pardo}},\ }\href@noop {} {\bibfield  {journal} {\bibinfo  {journal}
  {Proc. Natl. Acad. Sci. USA}\ }\textbf {\bibinfo {volume} {106}},\ \bibinfo
  {pages} {22073} (\bibinfo {year} {2009})}\BibitemShut {NoStop}%
\bibitem [{\citenamefont {Lu}\ and\ \citenamefont
  {Zhou}(2011)}]{lu2011prediction}%
  \BibitemOpen
  \bibfield  {author} {\bibinfo {author} {\bibfnamefont {L.}~\bibnamefont
  {Lu}}\ and\ \bibinfo {author} {\bibfnamefont {T.}~\bibnamefont {Zhou}},\
  }\href@noop {} {\bibfield  {journal} {\bibinfo  {journal} {Physica A: Stat.
  Mech. Appl.}\ }\textbf {\bibinfo {volume} {390}},\ \bibinfo {pages} {1150}
  (\bibinfo {year} {2011})}\BibitemShut {NoStop}%
\bibitem [{\citenamefont {Ghasemian}\ \emph {et~al.}(2019)\citenamefont
  {Ghasemian}, \citenamefont {Hosseinmardi},\ and\ \citenamefont
  {Clauset}}]{ghasemian2019overfitting}%
  \BibitemOpen
  \bibfield  {author} {\bibinfo {author} {\bibfnamefont {A.}~\bibnamefont
  {Ghasemian}}, \bibinfo {author} {\bibfnamefont {H.}~\bibnamefont
  {Hosseinmardi}},\ and\ \bibinfo {author} {\bibfnamefont {A.}~\bibnamefont
  {Clauset}},\ }\href@noop {} {\bibfield  {journal} {\bibinfo  {journal} {IEEE
  Trans. Knowl. Data Eng.}\ }\textbf {\bibinfo {volume} {1}},\ \bibinfo {pages}
  {1} (\bibinfo {year} {2019})}\BibitemShut {NoStop}%
\bibitem [{\citenamefont {Martin}\ \emph {et~al.}(2016)\citenamefont {Martin},
  \citenamefont {Ball},\ and\ \citenamefont
  {Newman}}]{martin2016reconstruction}%
  \BibitemOpen
  \bibfield  {author} {\bibinfo {author} {\bibfnamefont {T.}~\bibnamefont
  {Martin}}, \bibinfo {author} {\bibfnamefont {B.}~\bibnamefont {Ball}},\ and\
  \bibinfo {author} {\bibfnamefont {M.~E.~J.}\ \bibnamefont {Newman}},\
  }\href@noop {} {\bibfield  {journal} {\bibinfo  {journal} {Phys. Rev. E}\
  }\textbf {\bibinfo {volume} {93}},\ \bibinfo {pages} {012306} (\bibinfo
  {year} {2016})}\BibitemShut {NoStop}%
\bibitem [{\citenamefont
  {Newman}(2018{\natexlab{a}})}]{newman2018reconstruction}%
  \BibitemOpen
  \bibfield  {author} {\bibinfo {author} {\bibfnamefont {M.~E.~J.}\
  \bibnamefont {Newman}},\ }\href@noop {} {\bibfield  {journal} {\bibinfo
  {journal} {Nat. Phys.}\ }\textbf {\bibinfo {volume} {14}},\ \bibinfo {pages}
  {542} (\bibinfo {year} {2018}{\natexlab{a}})}\BibitemShut {NoStop}%
\bibitem [{\citenamefont
  {Newman}(2018{\natexlab{b}})}]{newman2018reconstruction2}%
  \BibitemOpen
  \bibfield  {author} {\bibinfo {author} {\bibfnamefont {M.~E.~J.}\
  \bibnamefont {Newman}},\ }\href@noop {} {\bibfield  {journal} {\bibinfo
  {journal} {Phys. Rev. E}\ }\textbf {\bibinfo {volume} {98}},\ \bibinfo
  {pages} {062321} (\bibinfo {year} {2018}{\natexlab{b}})}\BibitemShut
  {NoStop}%
\bibitem [{\citenamefont {Peixoto}(2018)}]{peixoto2018reconstruction}%
  \BibitemOpen
  \bibfield  {author} {\bibinfo {author} {\bibfnamefont {T.~P.}\ \bibnamefont
  {Peixoto}},\ }\href@noop {} {\bibfield  {journal} {\bibinfo  {journal} {Phys.
  Rev. X}\ }\textbf {\bibinfo {volume} {8}},\ \bibinfo {pages} {041011}
  (\bibinfo {year} {2018})}\BibitemShut {NoStop}%
\bibitem [{\citenamefont {Squartini}\ \emph {et~al.}(2018)\citenamefont
  {Squartini}, \citenamefont {Caldarelli}, \citenamefont {Cimini},
  \citenamefont {Gabrielli},\ and\ \citenamefont
  {Garlaschelli}}]{squartini2018reconstruction}%
  \BibitemOpen
  \bibfield  {author} {\bibinfo {author} {\bibfnamefont {T.}~\bibnamefont
  {Squartini}}, \bibinfo {author} {\bibfnamefont {G.}~\bibnamefont
  {Caldarelli}}, \bibinfo {author} {\bibfnamefont {G.}~\bibnamefont {Cimini}},
  \bibinfo {author} {\bibfnamefont {A.}~\bibnamefont {Gabrielli}},\ and\
  \bibinfo {author} {\bibfnamefont {D.}~\bibnamefont {Garlaschelli}},\
  }\href@noop {} {\bibfield  {journal} {\bibinfo  {journal} {Phys. Rep.}\
  }\textbf {\bibinfo {volume} {757}},\ \bibinfo {pages} {1} (\bibinfo {year}
  {2018})}\BibitemShut {NoStop}%
\bibitem [{\citenamefont {Shannon}(1948)}]{shannon1948mathematical}%
  \BibitemOpen
  \bibfield  {author} {\bibinfo {author} {\bibfnamefont {C.~E.}\ \bibnamefont
  {Shannon}},\ }\href@noop {} {\bibfield  {journal} {\bibinfo  {journal} {Bell
  Syst. Tech. J.}\ }\textbf {\bibinfo {volume} {27}},\ \bibinfo {pages} {379}
  (\bibinfo {year} {1948})}\BibitemShut {NoStop}%
\bibitem [{\citenamefont {Basharin}(1959)}]{basharin1959entropy}%
  \BibitemOpen
  \bibfield  {author} {\bibinfo {author} {\bibfnamefont {G.~P.}\ \bibnamefont
  {Basharin}},\ }\href@noop {} {\bibfield  {journal} {\bibinfo  {journal}
  {Theory Probab. Appl.}\ }\textbf {\bibinfo {volume} {4}},\ \bibinfo {pages}
  {333} (\bibinfo {year} {1959})}\BibitemShut {NoStop}%
\bibitem [{\citenamefont {Peixoto}(2014)}]{peixoto2014hierarchical}%
  \BibitemOpen
  \bibfield  {author} {\bibinfo {author} {\bibfnamefont {T.~P.}\ \bibnamefont
  {Peixoto}},\ }\href@noop {} {\bibfield  {journal} {\bibinfo  {journal} {Phys.
  Rev. X}\ }\textbf {\bibinfo {volume} {4}},\ \bibinfo {pages} {011047}
  (\bibinfo {year} {2014})}\BibitemShut {NoStop}%
\bibitem [{\citenamefont {Vall{\`e}s-Catal{\`a}}\ \emph
  {et~al.}(2018)\citenamefont {Vall{\`e}s-Catal{\`a}}, \citenamefont {Peixoto},
  \citenamefont {Sales-Pardo},\ and\ \citenamefont
  {Guimer{\`a}}}]{valles2018consistencies}%
  \BibitemOpen
  \bibfield  {author} {\bibinfo {author} {\bibfnamefont {T.}~\bibnamefont
  {Vall{\`e}s-Catal{\`a}}}, \bibinfo {author} {\bibfnamefont {T.~P.}\
  \bibnamefont {Peixoto}}, \bibinfo {author} {\bibfnamefont {M.}~\bibnamefont
  {Sales-Pardo}},\ and\ \bibinfo {author} {\bibfnamefont {R.}~\bibnamefont
  {Guimer{\`a}}},\ }\href@noop {} {\bibfield  {journal} {\bibinfo  {journal}
  {Phys. Rev. E}\ }\textbf {\bibinfo {volume} {97}},\ \bibinfo {pages} {062316}
  (\bibinfo {year} {2018})}\BibitemShut {NoStop}%
\bibitem [{\citenamefont {Wolpert}\ and\ \citenamefont
  {Wolf}(1995)}]{wolpert1995bayes}%
  \BibitemOpen
  \bibfield  {author} {\bibinfo {author} {\bibfnamefont {D.~H.}\ \bibnamefont
  {Wolpert}}\ and\ \bibinfo {author} {\bibfnamefont {D.~R.}\ \bibnamefont
  {Wolf}},\ }\href@noop {} {\bibfield  {journal} {\bibinfo  {journal} {Phys.
  Rev. E}\ }\textbf {\bibinfo {volume} {52}},\ \bibinfo {pages} {6841}
  (\bibinfo {year} {1995})}\BibitemShut {NoStop}%
\bibitem [{\citenamefont {Grassberger}(2008)}]{grassberger2008entropy}%
  \BibitemOpen
  \bibfield  {author} {\bibinfo {author} {\bibfnamefont {P.}~\bibnamefont
  {Grassberger}},\ }\href@noop {} {}\bibinfo {howpublished} {arXiv:0307138}
  (\bibinfo {year} {2008})\BibitemShut {NoStop}%
\bibitem [{\citenamefont {Peixoto}(2017)}]{peixoto2017dcsbm}%
  \BibitemOpen
  \bibfield  {author} {\bibinfo {author} {\bibfnamefont {T.~P.}\ \bibnamefont
  {Peixoto}},\ }\href@noop {} {\bibfield  {journal} {\bibinfo  {journal} {Phys.
  Rev. E}\ }\textbf {\bibinfo {volume} {95}},\ \bibinfo {pages} {012317}
  (\bibinfo {year} {2017})}\BibitemShut {NoStop}%
\bibitem [{\citenamefont {Peixoto}(2020)}]{peixoto2020mergesplit}%
  \BibitemOpen
  \bibfield  {author} {\bibinfo {author} {\bibfnamefont {T.~P.}\ \bibnamefont
  {Peixoto}},\ }\href@noop {} {}\bibinfo {howpublished} {arXiv:2003.07070}
  (\bibinfo {year} {2020})\BibitemShut {NoStop}%
\bibitem [{\citenamefont {Lancichinetti}\ and\ \citenamefont
  {Fortunato}(2009)}]{lancichinetti2009comparison}%
  \BibitemOpen
  \bibfield  {author} {\bibinfo {author} {\bibfnamefont {A.}~\bibnamefont
  {Lancichinetti}}\ and\ \bibinfo {author} {\bibfnamefont {S.}~\bibnamefont
  {Fortunato}},\ }\href@noop {} {\bibfield  {journal} {\bibinfo  {journal}
  {Phys. Rev. E}\ }\textbf {\bibinfo {volume} {80}},\ \bibinfo {pages} {056117}
  (\bibinfo {year} {2009})}\BibitemShut {NoStop}%
\bibitem [{\citenamefont {Miller}(1955)}]{miller1955entropy}%
  \BibitemOpen
  \bibfield  {author} {\bibinfo {author} {\bibfnamefont {G.}~\bibnamefont
  {Miller}},\ }in\ \href@noop {} {\emph {\bibinfo {booktitle} {Information
  Theory in Psychology; Problems and Methods}}},\ \bibinfo {editor} {edited by\
  \bibinfo {editor} {\bibfnamefont {H.}~\bibnamefont {Quastler}}}\ (\bibinfo
  {publisher} {Free Press},\ \bibinfo {address} {Glencoe, IL},\ \bibinfo {year}
  {1955})\BibitemShut {NoStop}%
\bibitem [{\citenamefont {Zahl}(1977)}]{zahl1977jackknife}%
  \BibitemOpen
  \bibfield  {author} {\bibinfo {author} {\bibfnamefont {S.}~\bibnamefont
  {Zahl}},\ }\href@noop {} {\bibfield  {journal} {\bibinfo  {journal}
  {Ecology}\ }\textbf {\bibinfo {volume} {58}},\ \bibinfo {pages} {907}
  (\bibinfo {year} {1977})}\BibitemShut {NoStop}%
\bibitem [{\citenamefont {Nemenman}\ \emph {et~al.}(2002)\citenamefont
  {Nemenman}, \citenamefont {Shafee},\ and\ \citenamefont
  {Bialek}}]{nemenman2002nsb}%
  \BibitemOpen
  \bibfield  {author} {\bibinfo {author} {\bibfnamefont {I.}~\bibnamefont
  {Nemenman}}, \bibinfo {author} {\bibfnamefont {F.}~\bibnamefont {Shafee}},\
  and\ \bibinfo {author} {\bibfnamefont {W.}~\bibnamefont {Bialek}},\ }in\
  \href@noop {} {\emph {\bibinfo {booktitle} {Advances in Neural Information
  Processing Systems 14}}}\ (\bibinfo  {publisher} {MIT Press},\ \bibinfo
  {address} {Cambridge, MA},\ \bibinfo {year} {2002})\BibitemShut {NoStop}%
\bibitem [{\citenamefont {Archer}\ \emph {et~al.}(2014)\citenamefont {Archer},
  \citenamefont {Park},\ and\ \citenamefont {Pillow}}]{archer2014pym}%
  \BibitemOpen
  \bibfield  {author} {\bibinfo {author} {\bibfnamefont {E.}~\bibnamefont
  {Archer}}, \bibinfo {author} {\bibfnamefont {I.~M.}\ \bibnamefont {Park}},\
  and\ \bibinfo {author} {\bibfnamefont {J.~W.}\ \bibnamefont {Pillow}},\
  }\href@noop {} {\bibfield  {journal} {\bibinfo  {journal} {J. Mach. Learn.
  Res.}\ }\textbf {\bibinfo {volume} {15}},\ \bibinfo {pages} {2833} (\bibinfo
  {year} {2014})}\BibitemShut {NoStop}%
\bibitem [{\citenamefont {Mitzenmacher}\ and\ \citenamefont
  {Upfal}(2005)}]{mitzenmacher2005probability}%
  \BibitemOpen
  \bibfield  {author} {\bibinfo {author} {\bibfnamefont {M.}~\bibnamefont
  {Mitzenmacher}}\ and\ \bibinfo {author} {\bibfnamefont {E.}~\bibnamefont
  {Upfal}},\ }\href@noop {} {\emph {\bibinfo {title} {{P}robability and
  {C}omputing: {R}andomized {A}lgorithms and {P}robabilistic {A}nalysis}}}\
  (\bibinfo  {publisher} {Cambridge University Press},\ \bibinfo {address} {New
  York, NY},\ \bibinfo {year} {2005})\BibitemShut {NoStop}%
\bibitem [{\citenamefont {Erd\H{o}s}\ and\ \citenamefont
  {R\'{e}nyi}(1959)}]{erdos1959randomgraph}%
  \BibitemOpen
  \bibfield  {author} {\bibinfo {author} {\bibfnamefont {P.}~\bibnamefont
  {Erd\H{o}s}}\ and\ \bibinfo {author} {\bibfnamefont {A.}~\bibnamefont
  {R\'{e}nyi}},\ }\href@noop {} {\bibfield  {journal} {\bibinfo  {journal}
  {Publ. Math. Debrecen}\ }\textbf {\bibinfo {volume} {6}},\ \bibinfo {pages}
  {290} (\bibinfo {year} {1959})}\BibitemShut {NoStop}%
\bibitem [{\citenamefont {Edler}\ \emph {et~al.}(2020)\citenamefont {Edler},
  \citenamefont {Eriksson},\ and\ \citenamefont {Rosvall}}]{infomap}%
  \BibitemOpen
  \bibfield  {author} {\bibinfo {author} {\bibfnamefont {D.}~\bibnamefont
  {Edler}}, \bibinfo {author} {\bibfnamefont {A.}~\bibnamefont {Eriksson}},\
  and\ \bibinfo {author} {\bibfnamefont {M.}~\bibnamefont {Rosvall}},\ }\href
  {https://www.mapequation.org} {\emph {\bibinfo {title} {The Infomap Software
  Package}}} (\bibinfo {year} {2020})\BibitemShut {NoStop}%
\bibitem [{\citenamefont {Rosvall}\ and\ \citenamefont
  {Bergstrom}(2011)}]{rosvall2011multilevelmapeq}%
  \BibitemOpen
  \bibfield  {author} {\bibinfo {author} {\bibfnamefont {M.}~\bibnamefont
  {Rosvall}}\ and\ \bibinfo {author} {\bibfnamefont {C.~T.}\ \bibnamefont
  {Bergstrom}},\ }\href@noop {} {\bibfield  {journal} {\bibinfo  {journal}
  {PLoS One}\ }\textbf {\bibinfo {volume} {6}},\ \bibinfo {pages} {e18209}
  (\bibinfo {year} {2011})}\BibitemShut {NoStop}%
\bibitem [{\citenamefont {Sch{\"u}rmann}(2004)}]{schurmann2004bias}%
  \BibitemOpen
  \bibfield  {author} {\bibinfo {author} {\bibfnamefont {T.}~\bibnamefont
  {Sch{\"u}rmann}},\ }\href@noop {} {\bibfield  {journal} {\bibinfo  {journal}
  {J. Phys. A}\ }\textbf {\bibinfo {volume} {37}},\ \bibinfo {pages} {L295}
  (\bibinfo {year} {2004})}\BibitemShut {NoStop}%
\bibitem [{\citenamefont {Gleiser}\ and\ \citenamefont
  {Danon}(2003)}]{gleiser2003jazznet}%
  \BibitemOpen
  \bibfield  {author} {\bibinfo {author} {\bibfnamefont {P.~M.}\ \bibnamefont
  {Gleiser}}\ and\ \bibinfo {author} {\bibfnamefont {L.}~\bibnamefont
  {Danon}},\ }\href@noop {} {\bibfield  {journal} {\bibinfo  {journal} {Adv.
  Complex Syst.}\ }\textbf {\bibinfo {volume} {06}},\ \bibinfo {pages} {565}
  (\bibinfo {year} {2003})}\BibitemShut {NoStop}%
\bibitem [{\citenamefont {Lancichinetti}\ \emph {et~al.}(2008)\citenamefont
  {Lancichinetti}, \citenamefont {Fortunato},\ and\ \citenamefont
  {Radicchi}}]{lancichinetti2008benchmark}%
  \BibitemOpen
  \bibfield  {author} {\bibinfo {author} {\bibfnamefont {A.}~\bibnamefont
  {Lancichinetti}}, \bibinfo {author} {\bibfnamefont {S.}~\bibnamefont
  {Fortunato}},\ and\ \bibinfo {author} {\bibfnamefont {F.}~\bibnamefont
  {Radicchi}},\ }\href@noop {} {\bibfield  {journal} {\bibinfo  {journal}
  {Phys. Rev. E}\ }\textbf {\bibinfo {volume} {78}},\ \bibinfo {pages} {046110}
  (\bibinfo {year} {2008})}\BibitemShut {NoStop}%
\bibitem [{\citenamefont {Newman}(2016)}]{newman2016resolution}%
  \BibitemOpen
  \bibfield  {author} {\bibinfo {author} {\bibfnamefont {M.~E.~J.}\
  \bibnamefont {Newman}},\ }\href@noop {} {\bibfield  {journal} {\bibinfo
  {journal} {Phys. Rev. E}\ }\textbf {\bibinfo {volume} {94}},\ \bibinfo
  {pages} {052315} (\bibinfo {year} {2016})}\BibitemShut {NoStop}%
\bibitem [{\citenamefont {Peel}\ \emph {et~al.}(2017)\citenamefont {Peel},
  \citenamefont {Larremore},\ and\ \citenamefont
  {Clauset}}]{peel2017groundtruth}%
  \BibitemOpen
  \bibfield  {author} {\bibinfo {author} {\bibfnamefont {L.}~\bibnamefont
  {Peel}}, \bibinfo {author} {\bibfnamefont {D.~B.}\ \bibnamefont
  {Larremore}},\ and\ \bibinfo {author} {\bibfnamefont {A.}~\bibnamefont
  {Clauset}},\ }\href@noop {} {\bibfield  {journal} {\bibinfo  {journal} {Sci.
  Adv.}\ }\textbf {\bibinfo {volume} {3}},\ \bibinfo {pages} {e1602548}
  (\bibinfo {year} {2017})}\BibitemShut {NoStop}%
\bibitem [{\citenamefont {Rosvall}\ and\ \citenamefont
  {Bergstrom}(2010)}]{rosvall2010bootstrap}%
  \BibitemOpen
  \bibfield  {author} {\bibinfo {author} {\bibfnamefont {M.}~\bibnamefont
  {Rosvall}}\ and\ \bibinfo {author} {\bibfnamefont {C.~T.}\ \bibnamefont
  {Bergstrom}},\ }\href@noop {} {\bibfield  {journal} {\bibinfo  {journal}
  {PLoS One}\ }\textbf {\bibinfo {volume} {5}},\ \bibinfo {pages} {1} (\bibinfo
  {year} {2010})}\BibitemShut {NoStop}%
\bibitem [{\citenamefont {Vinh}\ \emph {et~al.}(2010)\citenamefont {Vinh},
  \citenamefont {Epps},\ and\ \citenamefont {Bailey}}]{vinh2010ami}%
  \BibitemOpen
  \bibfield  {author} {\bibinfo {author} {\bibfnamefont {N.~X.}\ \bibnamefont
  {Vinh}}, \bibinfo {author} {\bibfnamefont {J.}~\bibnamefont {Epps}},\ and\
  \bibinfo {author} {\bibfnamefont {J.}~\bibnamefont {Bailey}},\ }\href@noop {}
  {\bibfield  {journal} {\bibinfo  {journal} {J. Mach. Learn. Res.}\ }\textbf
  {\bibinfo {volume} {11}},\ \bibinfo {pages} {2837} (\bibinfo {year}
  {2010})}\BibitemShut {NoStop}%
\bibitem [{\citenamefont {Leskovec}\ \emph {et~al.}(2007)\citenamefont
  {Leskovec}, \citenamefont {Kleinberg},\ and\ \citenamefont
  {Faloutsos}}]{leskovec2007arxiv}%
  \BibitemOpen
  \bibfield  {author} {\bibinfo {author} {\bibfnamefont {J.}~\bibnamefont
  {Leskovec}}, \bibinfo {author} {\bibfnamefont {J.}~\bibnamefont
  {Kleinberg}},\ and\ \bibinfo {author} {\bibfnamefont {C.}~\bibnamefont
  {Faloutsos}},\ }\href@noop {} {\bibfield  {journal} {\bibinfo  {journal} {ACM
  Trans. Knowl. Discovery Data}\ }\textbf {\bibinfo {volume} {1}},\ \bibinfo
  {pages} {2} (\bibinfo {year} {2007})}\BibitemShut {NoStop}%
\bibitem [{\citenamefont {Ebel}\ \emph {et~al.}(2002)\citenamefont {Ebel},
  \citenamefont {Mielsch},\ and\ \citenamefont {Bornholdt}}]{ebel2002email}%
  \BibitemOpen
  \bibfield  {author} {\bibinfo {author} {\bibfnamefont {H.}~\bibnamefont
  {Ebel}}, \bibinfo {author} {\bibfnamefont {L.-I.}\ \bibnamefont {Mielsch}},\
  and\ \bibinfo {author} {\bibfnamefont {S.}~\bibnamefont {Bornholdt}},\
  }\href@noop {} {\bibfield  {journal} {\bibinfo  {journal} {Phys. Rev. E}\
  }\textbf {\bibinfo {volume} {66}},\ \bibinfo {pages} {035103} (\bibinfo
  {year} {2002})}\BibitemShut {NoStop}%
\bibitem [{\citenamefont {Batagelj}\ and\ \citenamefont
  {Mrvar}(2000)}]{batagelj2000erdos}%
  \BibitemOpen
  \bibfield  {author} {\bibinfo {author} {\bibfnamefont {V.}~\bibnamefont
  {Batagelj}}\ and\ \bibinfo {author} {\bibfnamefont {A.}~\bibnamefont
  {Mrvar}},\ }\href@noop {} {\bibfield  {journal} {\bibinfo  {journal} {Soc.
  Networks}\ }\textbf {\bibinfo {volume} {22}},\ \bibinfo {pages} {173}
  (\bibinfo {year} {2000})}\BibitemShut {NoStop}%
\bibitem [{\citenamefont {Newman}\ and\ \citenamefont
  {Girvan}(2004)}]{newman2004football}%
  \BibitemOpen
  \bibfield  {author} {\bibinfo {author} {\bibfnamefont {M.~E.~J.}\
  \bibnamefont {Newman}}\ and\ \bibinfo {author} {\bibfnamefont
  {M.}~\bibnamefont {Girvan}},\ }\href@noop {} {\bibfield  {journal} {\bibinfo
  {journal} {Phys. Rev. E}\ }\textbf {\bibinfo {volume} {69}},\ \bibinfo
  {pages} {026113} (\bibinfo {year} {2004})}\BibitemShut {NoStop}%
\bibitem [{\citenamefont {Bogu\~n\'a}\ \emph {et~al.}(2004)\citenamefont
  {Bogu\~n\'a}, \citenamefont {Pastor-Satorras}, \citenamefont
  {D\'{\i}az-Guilera},\ and\ \citenamefont {Arenas}}]{boguna2004pgp}%
  \BibitemOpen
  \bibfield  {author} {\bibinfo {author} {\bibfnamefont {M.}~\bibnamefont
  {Bogu\~n\'a}}, \bibinfo {author} {\bibfnamefont {R.}~\bibnamefont
  {Pastor-Satorras}}, \bibinfo {author} {\bibfnamefont {A.}~\bibnamefont
  {D\'{\i}az-Guilera}},\ and\ \bibinfo {author} {\bibfnamefont
  {A.}~\bibnamefont {Arenas}},\ }\href@noop {} {\bibfield  {journal} {\bibinfo
  {journal} {Phys. Rev. E}\ }\textbf {\bibinfo {volume} {70}},\ \bibinfo
  {pages} {056122} (\bibinfo {year} {2004})}\BibitemShut {NoStop}%
\bibitem [{\citenamefont {Adamic}\ and\ \citenamefont
  {Glance}(2005)}]{adamic2005polblogs}%
  \BibitemOpen
  \bibfield  {author} {\bibinfo {author} {\bibfnamefont {L.~A.}\ \bibnamefont
  {Adamic}}\ and\ \bibinfo {author} {\bibfnamefont {N.}~\bibnamefont
  {Glance}},\ }in\ \href@noop {} {\emph {\bibinfo {booktitle} {Proceedings of
  the WWW-2005 Workshop on the Weblogging Ecosystem}}}\ (\bibinfo  {publisher}
  {ACM},\ \bibinfo {year} {2005})\BibitemShut {NoStop}%
\end{thebibliography}
\end{document}